\documentclass[a4paper,10pt]{article}

\usepackage{a4wide}
\usepackage[utf8]{inputenc}

\DeclareUnicodeCharacter{00A0}{~}

\usepackage{amssymb,amsmath,amscd,amsfonts,amsthm,bbm,mathrsfs,enumerate}
\usepackage{booktabs}
\usepackage[colorlinks=true, linkcolor=blue, citecolor=blue]{hyperref}
\usepackage{mathabx}
\usepackage{graphicx}
\usepackage{color}
\usepackage{accents}
\usepackage{subcaption}
\usepackage{enumitem}
\usepackage{authblk} 

\theoremstyle{definition}
\newtheorem{theorem}{Theorem}
\newtheorem*{theorem*}{Statement}
\newtheorem{lemma}[theorem]{Lemma}
\newtheorem{corollary}[theorem]{Corollary}
\newtheorem{proposition}[theorem]{Proposition}

\newtheorem{remark}{Remark}

\newtheorem*{condition*}{Condition}
\newtheorem{assumption}{Assumption}

\DeclareMathOperator{\PL}{PL}
\DeclareMathOperator*{\diag}{diag}
\DeclareMathOperator{\Tanh}{Tanh}

\newcommand{\1}{\mathbbm 1}
\newcommand{\T}{\top}

\newcommand{\PP}{{{\mathbb P}}} 
\newcommand{\EE}{{{\mathbb E}}} 
\newcommand{\NN}{{{\mathbb N}}} 
\newcommand{\RR}{{{\mathbb R}}}

\newcommand{\mcF}{{\mathscr F}} 
 
\newcommand{\mcL}{{\mathscr L}}

\newcommand{\cA}{{\mathcal A}} 
\newcommand{\cB}{{\mathcal B}} 
\newcommand{\cC}{{\mathcal C}} 
\newcommand{\cD}{{\mathcal D}} 
\newcommand{\cE}{{\mathcal E}} 
 
\newcommand{\cG}{{\mathcal G}}

\newcommand{\cK}{{\mathcal K}}

\newcommand{\cN}{{\mathcal N}} 
\newcommand{\cO}{{\mathcal O}} 
\newcommand{\cP}{{\mathcal P}} 
\newcommand{\cR}{{\mathcal R}} 
 
\newcommand{\cS}{{\mathcal S}} 
\newcommand{\cT}{{\mathcal T}}

\newcommand{\Ccard}{\bC_{\text{card}}} 
\newcommand{\Scirc}{S_{\text{circ}}}

\newcommand{\ton}{\xrightarrow[n\to\infty]{}} 

\newcommand{\ps}[2]{\left(#1 \cdot #2 \right)}

\newcommand{\EQ}[1]{\left\langle #1 \right\rangle}

\newcommand{\rownorm}[1]{\left\lvert\hspace{-1 pt}\left\lvert\hspace{-1 pt}\left\lvert {#1}\right\lvert\hspace{-1 pt}\right\lvert\hspace{-1 pt}\right\lvert}

\newcommand{\bC}{\bs C}

\newcommand{\tx}{\tilde x} 
 
\newcommand{\cq}{\check q} 
\newcommand{\cx}{\check x} 
\newcommand{\cX}{{\widecheck X}} 
 
\newcommand{\cz}{\check z} 
\newcommand{\Ac}{A^{\text{c}}} 

\newcommand{\bs}{\boldsymbol}

\begin{document}


\title{The SK model with a sparse variance profile: free energy
and AMP-based magnetization estimation at high temperature}



\author{Walid Hachem}
\affil{CNRS, Laboratoire d’informatique Gaspard Monge (LIGM / UMR 8049),
Université Gustave Eiffel, ESIEE Paris, France \\
\texttt{walid.hachem@univ-eiffel.fr} }


\date{\today}

\maketitle 

\begin{abstract}
A generalization of the Sherrington–Kirkpatrick (SK) model for spin glasses is
considered, in which the interaction matrix is endowed with a variance profile
that has no particular structure and may be sparse. In the first part of the
paper, an asymptotic equivalent of the free energy is derived at sufficiently
high temperatures, regardless of the signature of the variance profile matrix.
In the second part, the mean of the spin vector under the Gibbs measure is
estimated using an  Approximate Message Passing algorithm that generalizes
the one developed for the classical SK case. 
The dynamical approach of Adhikari
\emph{et al.} (J. Stat. Phys., 2021), originally developed for the classical
SK model, is adapted to the present setting to obtain these results.
\end{abstract}

\section{Model, problem and the results} 

For each integer $n > 0$, let $S^{(n)} = \begin{bmatrix} s^{(n)}_{ij}
\end{bmatrix}_{i,j=1}^n$ be a deterministic symmetric matrix with elements
$s^{(n)}_{ij} \geq 0$ and with a zero diagonal. Let $W^{(n)} = \begin{bmatrix}
W^{(n)}_{ij} \end{bmatrix}_{i,j=1}^n$ be a real symmetric $n\times n$ random
matrix such that the random variables $\{ W^{(n)}_{ij} \}_{1\leq i < j \leq n}$
are independent, and such that $W^{(n)}_{ij} \sim \cN(0, t s_{ij}^{(n)})$ for 
all $i,j \in [n]$ and for some $t > 0$. 

Let $\Sigma_n = \{-1,+1\}^n$ be the space of vectors of Ising spins with size
$n$. Define the random $\cP(\Sigma_n)$--valued Gibbs measure $G^{(n)}$ as
follows. The measure of a singleton $\{\sigma\} \subset \Sigma_n$ by $G^{(n)}$
is $G^{(n)}(\sigma) = \exp(H^{(n)}(\sigma)) / Z^{(n)}$, where $H^{(n)}$ is
the Hamiltonian defined as 
\[
H^{(n)}(\sigma) = \frac 12 \sigma^\T W^{(n)} \sigma + 
  h \ps{\sigma}{1_n} , 
\]
$h\in\RR$ is the amplitude of an external field, and $Z^{(n)} =
\sum_{\sigma\in\Sigma_n} \exp(H^{(n)}(\sigma))$ is the partition function.  In
the particular case where $s_{ij}^{(n)} = 1/n$ for $i\neq j$, this model for the
Gibbs measure boils down to the classical Sherrington-Kirkpatrick (SK) model
that has been studied at length in the statistical physics literature, as
detailed in the treatises \cite{tal-livre11-t1,tal-livre11-t2,pan-livre13}.
Since the interactions between our spins are subjected to the more general
variance profile represented by the matrix $S^{(n)}$, we term our model a SK 
model with a variance profile.

Let us state our conditions on this matrix. Considering a sequence of 
positive numbers $(K_n)$ such that $K_n\leq n$ and $K_n \to \infty$, we assume 
the following: 
\begin{assumption} 
\label{ass-K}
It holds that 
\begin{itemize}
\item There exists a constant $ \bs C_s > 0$ such that
 $s_{ij}^{(n)} \leq \bs C_s K_n^{-1}$ for all $n$ and all $i,j \in [n]$.
\item The number $\bC_{\text{row}} = \sup_n \rownorm{S^{(n)}}$ where 
$\rownorm{\cdot}$ is the max row sum norm is finite. 
\end{itemize}
\end{assumption}
One case of interest covered by Assumption~\ref{ass-K} and that can be useful
in the fields of statistical physics, graph inference, and large-dimensional
signal estimation among others is the case where $K_n \ll n$, and where
there are at most $\Ccard K_n$ non-zero $s_{ij}^{(n)}$'s in each row and column
of $S^{(n)}$ for some constant $\Ccard > 0$. We refer to these models as the 
``sparse'' ones. 

The first aim of this paper is to study the large $n$ asymptotics of the free 
energy 
\[
F_n = \frac 1n \EE\log Z^{(n)}
\]
in a ``high temperature'' regime represented by the following assumption:  
\begin{assumption} 
\label{ass-HT} 
$\displaystyle{t < \frac{\log 2}{\bC_{\text{row}}}}$. 
\end{assumption}

Denoting as usual as $\EQ{\cdot}$ the mean operator with respect to the measure
$(G^{(n)})^{\otimes\infty}$, the second aim of this paper is to provide a
large--$n$ approximation of the vector of magnetizations $m^{(n)} =
\EQ{\sigma}$.  Returning to the classical SK model, it is well-known that this
vector can be approximated in the high temperature regime by an Approximate
Message Passing (AMP) algorithm by building on the so-called
Thouless-Anderson-Palmer (TAP) equations~\cite{bol-14}.  The second purpose of
this paper is to generalize this result to our SK model with a variance profile
in the temperature regime specified by Assumption~\ref{ass-HT}. \\

The conditions on the variance profile that can be found in the mathematical
physics literature are usually much more restrictive than those provided by
Assumption~\ref{ass-K}. Regarding the free energy computation, this literature
is mainly limited to the so-called multi-species model, where it is assumed
that $S^{(n)}$ consists in a finite number of blocks whose dimensions scale
with $n$ (implying that $K_n = n$ in our setting), and typically, that
$S^{(n)}$ is a non-negative matrix in the semi-definite positive ordering.
These two assumptions are made in, \emph{e.g.},
\cite{bar-con-min-tan-15,pan-15} which deal, on the other hand, with the free
energy problem at all temperatures.  More recent contributions dealing with the
multi-species model with the non-negativity assumption include
\cite{bat-slo-soh-19,alb-cam-con-min-(arxiv)20,che-iss-mou-(arxiv)25,
kim-(arxiv)25}.  
These references invoke the non-negativity assumption on $S^{(n)}$ so that
Guerra's well-known interpolation technique applies. 
Cases where this matrix can be indefinite were discussed in
\cite{dey-wu-21,che-(arxiv)24, wu-24,bat-soh-(arxiv)25}.  Diluted variants of
the multi-species model are considered in \cite{liu-don-24,
alb-con-min-zim-(arxiv)24}.  Considering the potential applications, graph max
$\kappa$--cut problems and low-rank matrix estimation problems related to the
multi-species model were considered in \cite{jag-ko-sen-18} and in
\cite{gui-ko-krz-zde-25} respectively. It would be useful to extend these
results to more flexible and more general variance profiles such as the ones
considered in this paper. \\

In all this paper, we denote as $\xi$ a standard Gaussian random variable, and
we write $\Tanh(x) = \tanh(x+h)$. Define the scalar function $g:\RR_+\to\RR_+$ 
as 
\begin{equation} 
\label{def-g} 
g(x) = \EE \Tanh\left(\sqrt{x} \xi \right)^2 ,  
\end{equation} 
and consider the system of equations defined in $q^{(n)} \in \RR_+^n$ as 
\begin{equation}
\label{eq-q} 
q^{(n)} = t S^{(n)} g(q^{(n)}) 
\end{equation} 
where $g(q^{(n)})\in\RR_+^n$ is the vector obtained by an element-wise 
application of the function $g$ to the vector $q^{(n)}$ (this notational 
convention regarding scalar functions applied to vectors will be used all 
along this paper). 
\begin{lemma}
\label{lm-q} 
Under Assumptions~\ref{ass-K} and~\ref{ass-HT}, Equation~\eqref{eq-q} admits an
unique solution $q^{(n)} \in\RR_+^n$. Given any vector $q^{(n), 0} \in
\RR_+^n$, the iterative algorithm $q^{(n),l+1} = t S^{(n)} g( q^{(n),l})$
converges to this solution. Moreover, $\sup_n \| q^{(n)} \|_\infty < \log 2$ 
where $\|\cdot\|_\infty$ is the $\max$--norm. 
\end{lemma} 
This simple lemma is proven in Appendix~\ref{prf-lm-q} for completeness. The 
large--$n$ behavior of the free energy is specified by the following theorem: 
\begin{theorem}
\label{th-F} 
Let Assumptions \ref{ass-K} and \ref{ass-HT} hold true. 
Defining the function 
\[
\bs F_n = \log 2 + \frac 1n \sum_{i=1}^n 
 \EE\log\cosh\left( \sqrt{q^{(n)}_i} \xi + h \right)
 + \frac{t}{4n} \left(1_n-g(q^{(n)})\right)^\T S^{(n)}
    \left(1_n-g(q^{(n)})\right) , 
\]
where $q^{(n)} = \begin{bmatrix} q^{(n)}_i \end{bmatrix}_{i=1}^n$ is the 
solution of~\eqref{eq-q}, it holds that $\bs F_n$ is bounded, and moreover,
that 
\[
F_n - \bs F_n \ton  0 . 
\]
\end{theorem}
It is seen here that the signature of $S^{(n)}$ has no impact on the form of 
the large-$n$ approximation of the free energy. 

It is worth considering the particular case of this theorem where the symmetric
matrix $S^{(n)}$ is doubly stochastic, \emph{i.e.}, it satisfies $S^{(n)} 1_n =
1_n$.  In this case, the solution of Equation~\eqref{eq-q} is reduced to
$q^{(n)} = q 1_n$ where the scalar $q$ is the unique solution to $q = t g(q)$.
Therefore, as long as the degree of sparsity satisfies $K_n\to\infty$, we
recover the expression of the asymptotic free energy of the classical SK model
at high temperature \cite{tal-livre11-t1}:
\begin{corollary}[the doubly stochastic case] 
\label{cor-sto}
Assume that $S^{(n)}$ is a doubly stochastic matrix. Then, it holds under 
Assumptions~\ref{ass-K} and~\ref{ass-HT} (which reads $t< \log 2$) that 
\begin{equation}
\label{bisto} 
F_n \xrightarrow[n\to\infty]{}  
 \log 2 + \EE\log\cosh\left( \sqrt{q} \xi + h \right)
 + \frac{t}{4} (1-q/t)^2, 
\end{equation} 
where $q$ is the unique solution to the scalar equation $q = t g(q)$ on 
$\RR_+$. 
\end{corollary} 

A particular case that can be useful in the field of statistical physics is the
case where $S^{(n)}$ is a Toeplitz banded matrix with a bandwidth of order
$K_n\to\infty$ with $K_n\ll n$. Here, $K_n$ represents the range of the
interactions among the spins.  Within this range, the random interactions are
furthermore subjected to a variance profile that depends on the distance
between the sites as shown in the statement of the following corollary: 
\begin{corollary}[banded Toeplitz interaction profile] 
\label{cor-band}
Let $\psi^{(n)} : \{0, \ldots, n-1 \} \to \RR_+$ be a function that satisfies
the following assumptions: $\psi^{(n)}(i) \leq \bs C_s/K_n$, $\sum_i \psi^{(n)}(i)
= 1/2$, and the support of $\psi^{(n)}$ is included into $\{1, \ldots, K_n \}$.
Let $s^{(n)}_{ij} = \psi^{(n)}(|i-j|)$, and assume that $K_n\to\infty$ with
$K_n/n\to 0$. Then, for $t < \log 2$, the convergence~\eqref{bisto} holds true. 
\end{corollary}
The proof is provided in Appendix~\ref{prf-cor-band}. \\ 

We now tackle the approximation problem of the magnetization vector $m^{(n)} =
[ m^{(n)}_i ] = \EQ{\sigma}$ with the help of an AMP algorithm.  The well-known
idea underlying this algorithm for the classical SK model, 
where $S^{(n)} = n^{-1} 1_n 1_n^\T$, can be summarized as follows.  Using the 
cavity method developed in, \emph{e.g.}, \cite[Chap.~1]{tal-livre11-t1}, one 
can write 
\begin{equation}
\label{cavity} 
m_i^{(n)} \simeq \Tanh\left( \left[ W^{(n)} m^{(n)}_{(i)} \right]_i  
  \right) 
\end{equation}
for $t$ sufficiently small, where $m_{(i)}^{(n)} \in \RR^n$ is the vector of
magnetizations obtained by dropping the interactions with the spin $i$  (a
precise definition of this vector will be provided below). These approximations
lead to the TAP equations 
\[
m_i^{(n)} \simeq \Tanh\left( \left[ W^{(n)} m^{(n)} \right]_i   
 - t( 1-q) m^{(n)}_i 
  \right) 
\]
where $q$ is the solution to $q = t g(q)$. 
The TAP equations provide the foundation for the AMP algorithm introduced in
the seminal paper~\cite{bol-14} in the mathematical physics literature. 

To generalize these results to our interaction matrix with a variance profile,
we require a slightly more restrictive version of Assumption~\ref{ass-K}.
Keeping our sequence $K_n\to\infty$ with $K_n\leq n$, we set: 
\begin{assumption} The following facts hold true.
\label{ass-card}
\begin{itemize}
\item There exists a constant $ \bs C_s > 0$ such that
 $s_{ij}^{(n)} \leq \bs C_s K_n^{-1}$ for all $n$ and all $i,j \in [n]$.
\item There exists a constant $\Ccard > 0$ such that
\[
\forall n, \ \forall i \in [n], \ 
 \left| \left\{ j \in [n] \, : \, s_{ij}^{(n)} > 0 \right\} \right|
    \leq \Ccard K_n , 
\]
where $|\cdot|$ is the cardinality of a set.
\end{itemize}
\end{assumption}
Of course, Assumption~\ref{ass-card} implies Assumption~\ref{ass-K} with 
$\bC_{\text{row}} = \sup_n \rownorm{S^{(n)}}$ satisfying 
$\bC_{\text{row}} \leq \bs C_s \Ccard$. If $K_n \ll n$, 
Assumption~\ref{ass-card} models the sparse cases alluded to above.

In the remainder, we denote as $\|\cdot\|$ the Euclidean norm of a vector or
the spectral norm of a matrix. We also write $\|\cdot\|_n =
\|\cdot\|/\sqrt{n}$. 
Our analogue of Bolthausen’s result, which we shall prove by a technique
different from the one sketched above, reads as follows.
\begin{theorem}
\label{th-amp} 
Let Assumptions \ref{ass-card} and \ref{ass-HT} hold true. Assume that
$K_n \geq \log n$. Consider the iterates $(q^{(n),l})_{l\in\NN}$ defined in
the statement of Lemma~\ref{lm-q} starting with $q^{(n),0} = 0_n$. 
For each $l=0,1,\ldots$, let $X^{(n),l} \sim \cN(0, \diag(q^{(n),l}))$. 
Starting with $x^{(n),0} = 0$ and $x^{(n),1} = W^{(n)} \Tanh(0)$, 
consider the following iterative AMP algorithm in $l=1,\ldots$ 
\begin{align} 
x^{(n),l+1} &= W^{(n)} \Tanh\left(x^{(n),l} \right) 
 - \diag\left( t S^{(n)} 1_n - q^{(n),l+1} \right) 
  \Tanh\left(x^{(n),l-1} \right) \nonumber \\
 &= W^{(n)} \Tanh\left(x^{(n),l} \right) 
 - \diag\left( t S^{(n)} \EE\Tanh'\left(X^{(n),l}\right) \right) 
  \Tanh\left(x^{(n),l-1} \right) . 
\label{amp-x} 
\end{align} 
Then, the vector $m^{(n)} = \EQ{\sigma}$ satisfies 
\[
\lim_{k\to\infty} \limsup_n 
\EE \left\| m^{(n)} - \Tanh\left(x^{(n),k} \right) \right\|_n^2 = 0. 
\]
\end{theorem} 

We notice that when $S^{(n)} = n^{-1} 1_n 1_n^\T$, the algorithm boils down
to 
\[
x^{(n),l+1} = W^{(n)} \Tanh\left(x^{(n),l} \right) 
 - ( t - q_{l+1} ) \Tanh\left(x^{(n),l-1} \right) 
\]
where the scalar sequence $(q_l)$ is defined by the recursion 
$q_0 = 0$ and $q_{l+1} = t g(q_l)$. We observe that the term 
 $- ( t - q_{l+1} ) \Tanh(x^{(n),l-1})$ is a dynamical version of the 
so-called Onsager term  $- t( 1-q) m^{(n)}$ that appears at the right
hand side of the TAP equation above.

\section{Proofs} 

In the remainder, we omit the superscript $^{(n)}$ from the notations unless
when useful. We denote as $C > 0$ a constant that depends on
$\bC_{\text{row}}$, $\bs C_s$, $\bs C_{\text{card}}$ and $t$ at most, and that 
can change from one display to another. 

We shall approximate the free energy at high temperature via the 
Guerra's interpolation of the Hamiltonian. In our setting, this gives the
following scheme.  Let 
\[
\eta = \begin{bmatrix} \eta_i \end{bmatrix}_{i=1}^n 
  \sim \cN\left(0, \diag(q)\right)
\]
be independent from $W$, where we recall that $q$ is the unique solution of the
system $q = tSg(q)$ as shown by Lemma~\ref{lm-q}. Given $u\in[0,1]$, define 
the Hamiltonian on the space of spins $\Sigma_n$ as 
\[
H_u(\sigma) = \frac{\sqrt{u}}{2} \sigma^\T W \sigma + 
  h \ps{\sigma}{1} + \sqrt{1-u} \ps{\sigma}{\eta} , 
\]
and consider the Gibbs measure $G_u$ defined on $\Sigma_n$ as 
$G_u(\sigma) = \exp(H_u(\sigma)) / Z_u$ where 
$Z_u = \sum_\sigma \exp(H_u(\sigma))$ is the partition function. Our Gibbs 
measure of interest is of course $G_1$, and its free energy is 
$F_n = n^{-1} \EE\log Z_1$. 
The Hamiltonian $H_u$ is an interpolation between the Hamiltonian of interest 
$H_1$ and the Hamiltonian $H_0$ whose free energy has a tractable expression. 

To pursue this strategy, we introduce some notations.  Given a set of indices
$A\subset[n]$, we write $\Ac = [n]\setminus A$.  Given a vector $x = [ x_i ]\in
\RR^n$ we denote as $x_A\in\RR^n$ the vector $x$ which elements $x_i$ are set
to zero when $i\in\Ac$.  

We denote as $H_{u,(A)}$ the reduced Hamiltonian obtained by setting the spins
$\{ \sigma_i \}_{i\in A}$ to zero. More precisely, this Hamiltonian is written
as 
\[
H_{u,(A)}(\sigma_{\Ac}) = 
 \frac{\sqrt{u}}{2} \sigma_{\Ac}^\T W \sigma_{\Ac} + 
 h \ps{\sigma_{\Ac}}{1_n} + \sqrt{1-u} \ps{\sigma_{\Ac}}{\eta_{\Ac}} , 
\]
and can be considered as a Hamiltonian on $\{-1,+1\}^{|\Ac|}$. We denote as 
\[
G_{u,(A)}(\sigma_{\Ac}) \propto \exp( H_{u,(A)}(\sigma_{\Ac}))
\]
the Gibbs probability measure on $\{-1,+1\}^{|\Ac|}$ which Hamiltonian is
$H_{u,(A)}(\sigma_{\Ac})$. For $i\in\Ac$, we also denote as $m_{(A),i}$ the
mean of the spin $\sigma_i$ with respect to $G_{u,(A)}$. Conventionally, we
write $m_{(A),i} = 0$ when $i\in A$ so that we can define the vector $m_{(A)} =
[ m_{(A),i}]_{i\in[n]} \in \RR^n$.  When $A = \{ i_1, \ldots, i_k \}$, we
sometimes write $m_{(i_1,\ldots,i_k)} = [ m_{(i_1,\ldots,i_k),i} ]_i$ for the
vector  $m_{(A)} = [ m_{(A),i} ]_i$. Of course, $m_{(\emptyset)} = m = [m_i]$.

\paragraph*{General organization of the proof.} 
Our starting point will be the approach of Adhikari \emph{et.al.}~in
\cite{adh-bre-vSo-yau-21}, where the SK case with $u=1$ was
considered.  This approach falls within a research axis that dates back up to
our knowledge to the work of Comets and Neveu \cite{com-nev-95}, and that 
considers the SK model from a stochastic calculus perspective.  Denoting as 
$\EQ{\cdot}_u$ the mean with respect to the measure $G_u^{\otimes\infty}$, the 
first step is to show that the covariances 
$m_{ij} = \EQ{(\sigma_i - \EQ{\sigma_i}_u) (\sigma_j - \EQ{\sigma_j}_u)}_u$ 
for $i\neq j$ satisfy $\EE m_{ij}^2 \sim 1/K_n$. This is shown in 
Lemma~\ref{lm-var} below, which is a straightforward adaptation of 
\cite{adh-bre-vSo-yau-21} to the variance profile case of interest in this 
paper. 

Using this result, the ``pre-TAP'' bound 
\begin{equation} 
\label{bndpreTAP} 
\EE \left( m_i - \Tanh\left( 
 \sqrt{u} \sum_{k} W_{ik} m_{(i),k} + \sqrt{1-u} \eta_i \right) 
  \right)^2 \leq \frac{C}{K_n}, 
\end{equation} 
as well as the bound 
\[
\EE \left( m_l - m_{(i), l} \right)^2 \leq \frac{C}{K_n} , \quad l\neq i 
\]
can be obtained. These bounds are generalizations to our model of quite
well-known results in the SK literature, see, \emph{e.g.}, \cite[Lemma
1.7.4]{tal-livre11-t1}, or the approximation~\eqref{cavity} above. Here, they
will serve two purposes. 

First, defining the random vector 
\[
R_{12} = \begin{bmatrix} R_{12}(i) \end{bmatrix}_{i=1}^n 
 = t S (\sigma^1 \sigma^2) ,
\]
where $\sigma^1\sigma^2$ is the vector obtained by an elementwise product of
the elements of the replicas $\sigma^1$ and $\sigma^2$, it can be deduced from
these bounds that $\EE\EQ{\|R_{12} - q\|_n^2}_u \leq C K_n^{-1/2}$
(Lemma~\ref{lm-R12} below). Employing the usual Guerra's interpolation trick in
order to compute the free energy of $G_1$, we can see that thanks to the bound
$\EE\EQ{\|R_{12} - q\|_n^2}_u \leq C K_n^{-1/2}$, the ``annoying'' term
obtained through this interpolation is negligible, which leads to
Theorem~\ref{th-F}. We note here that in the references cited above which deal
with the multi-species model, it is assumed that $S$ is a non-negative matrix
specifically to force this term to be non-positive. 

Second, the bound \eqref{bndpreTAP} taken for $u=1$ leads to the iterative
construction~\eqref{mi1} for approximating $m_i$, as was done in Chen and Tang
in~\cite{che-tan-21} for the SK model.  Starting from this construction, we
shall devise a series of approximations of the vector $m = [m_i]$ that will
ultimately lead to the AMP approximation given by Theorem~\ref{th-amp}.  The
main steps of this procedure stand as follows.  By modifying the initialization
and replacing $\Tanh$ with a polynomial, we obtain the iterates
$\cx^l_{(\cdots)}$ listed in the table below. These admit a combinatorial
description in terms of trees in which the weights $W_{ij}$ along each path are
independent. Next, this independence along paths is relaxed to the weaker
``non-backtracking property'', which yields the iterates $\cz^l_{(\cdot)}$. We
then show that these iterates can in turn be replaced with the AMP iterates
$z^l$ associated with polynomial activation functions. These last two
approximations follow the approach of Bayati \emph{et
al.}~\cite{bay-lel-mon-15} for the classical AMP algorithm, as extended
in~\cite{hac-24} to AMP with an interaction matrix having a variance profile.
Finally, the polynomial activation functions are replaced with $\Tanh$, which
leads to Theorem~\ref{th-amp}.

\subsection{Adapting the approach of \cite{adh-bre-vSo-yau-21} to the
SK model with a variance profile} 

We need to introduce some more notations. For a set  $A \subset [n]$, we need
to work on the conditional interpolated Gibbs measure $G_u(\cdot | \sigma_A)$
given $\sigma_A$, which is the measure on $\{-1,+1\}^{|\Ac|}$ with the
Hamiltonian $\sigma_{\Ac} \mapsto H^{[A]}_u(\sigma_{\Ac})(\sigma_A)$ given as 
\[
H^{[A]}_u(\sigma_{\Ac})(\sigma_A) = 
 \frac{\sqrt{u}}{2} \sigma_{\Ac}^\T W \sigma_{\Ac} + 
 h \ps{\sigma_{\Ac}}{1} + \sqrt{1-u} \ps{\sigma_{\Ac}}{\eta_{\Ac}} 
 + \sqrt{u} \sigma_{\Ac}^\T W \sigma_{A} 
\]  
(recall that $\sigma_A$ and $\sigma_{\Ac}$ are $n\times 1$ vectors
with zeros at the appropriate coordinates).  Given $i,j\in \Ac$, we denote as
$m_i^{[A]}(\sigma_A)$ and $m_{ij}^{[A]}(\sigma_A)$ the mean of $\sigma_i$ and
the covariance of $\sigma_i$ and $\sigma_j$ with respect to the conditional
probability $G_u(\cdot | \sigma_A)$.  We also write $m_i^{[i_1,\ldots,i_k]}$
and $m_{ij}^{[i_1,\ldots,i_k]}$ for $m_i^{[A]}$ and $m_{ij}^{[A]}$ respectively
when $A = \{i_1,\ldots, i_k \}$.  For $i\in A$ and a real function
$f(\sigma_A)$, we also define the functions $\sigma_{A\setminus\{i\}} \mapsto
\delta_i f(\sigma_{A\setminus\{i\}})$ and $\sigma_{A\setminus\{i\}} \mapsto
\varepsilon_i f(\sigma_{A\setminus\{i\}})$ as 
\[
\delta_i f(\sigma_{A\setminus\{i\}}) = 
  \frac 12\left( f(\sigma_{A})_{| \sigma_i = 1}  
  - f(\sigma_{A})_{| \sigma_i = -1} \right) , 
\]
and 
\[
\varepsilon_i f(\sigma_{A\setminus\{i\}}) = 
  \frac 12\left( f(\sigma_{A})_{| \sigma_i = 1}  
  + f(\sigma_{A})_{| \sigma_i = -1} \right) . 
\]
The following key identity can be obtained by direct calculation and is 
provided in \cite[Eq.~(3.1)]{adh-bre-vSo-yau-21}:  
\begin{equation}
\label{mij-mi} 
m_{ij}^{[A]} = \left( 1 - \left( m_i^{[A]}\right)^2\right) 
  \delta_i m_j^{[A\cup\{i\}]} 
\end{equation} 
for $i,j\in \Ac$ with $i\neq j$. Following \cite{adh-bre-vSo-yau-21}, we focus
on the vector $[ \sqrt{u} W_{ik} ]_{k\not\in A\cup \{i\}}$ and we consider the
random variable $\sqrt{u} W_{ik}$ within this vector as the value at $tu$ of
the process $ \sqrt{s_{ik}} B_{ik}(v)$ where $[ B_{ik} ]_{k \not\in A\cup \{i\}
}$ is a standard multidimensional Brownian Motion. With the help of Itô's
lemma, it is possible to obtain a characterization of $\delta_i m_j^{[A\cup
\{i\}]}(\sigma_A)$.  Similarly to \cite[Eq.~(3.3)]{adh-bre-vSo-yau-21}, we
obtain by this lemma 
\begin{align*} 
\delta_i m_j^{[A\cup \{i\}]}(\sigma_A) &= 
 \sum_{k\not\in A\cup \{i\}} 
  \sqrt{s_{ik}} 
  \int_0^{tu} \varepsilon_i m_{kj}^{[A\cup \{i\}]}(\sigma_A)(v) dB_{ik}(v)  \\
 &\phantom{=} - \sum_{k\not\in A\cup \{i\}} s_{ik} 
  \int_0^{tu} 
 \delta_i \left( m_{k}^{[A\cup \{i\}]} m_{kj}^{[A\cup \{i\}]}\right)(\sigma_A)(v) dv .
\end{align*} 
(here, $\varepsilon_i m_{kj}^{[A\cup \{i\}]}(\sigma_A)(v)$ is of course the 
value of $\varepsilon_i m_{kj}^{[A\cup \{i\}]}(\sigma_A)$ for which 
$\sqrt{u} W_{ij}$ in the Hamiltonian is replaced with 
$\sqrt{s_{ij}} B_{ij}(v)$, and similarly for the second integrand). 

This Itô characterization of $\delta_i m_j^{[A\cup \{i\}]}$ together with the
identity~\eqref{mij-mi} lie at the basis of proof of the following result,
which is an adaptation of \cite[Lemma 3.1]{adh-bre-vSo-yau-21} to our
situation.  For completeness, we provide this proof in
Appendix~\ref{prf-lm-var}. 
\begin{lemma}
\label{lm-var} 
For $t\in[0, \log 2/\bC_{\text{row}})$, there exists a constant $C$ such that
\[
 \EE m_{ij}^2 \leq \frac{C}{K_n} \quad \text{for all } i\neq j . 
\]
\end{lemma}  

With the help of the previous lemma, we obtain the following result by 
a straightforward adaptation of the proof of 
\cite[Lemma 4.1]{adh-bre-vSo-yau-21}:
\begin{lemma}
\label{preTAP} 
For $t\in[0, \log 2/\bC_{\text{row}})$, there exists a constant $C$ such that 
\[
\EE \left( m_i - \Tanh\left( 
 \sqrt{u} \sum_{k\neq i} W_{ik} m_{(i),k} + \sqrt{1-u} \eta_i \right) 
  \right)^2 \leq \frac{C}{K_n}, 
\]
and 
\[
\EE \left( m_l - m_{(i), l} \right)^2 \leq \frac{C}{K_n} 
\]
for each $i \neq l \in [n]$. 
\end{lemma} 

Building the random vector $R_{12} = tS (\sigma^2\sigma^2)$ from two 
i.i.d.~vectors $\sigma^1$ and $\sigma^2$ under $G_u$ (the so-called replicas), 
we can now use this result to show that $\EE\EQ{\| R_{12} - q \|^2_n}_u$ 
converges to zero uniformly in $u\in [0,1]$: 
\begin{lemma}
\label{lm-R12} 
For $t\in[0, \log 2/\bC_{\text{row}})$, there exists a constant $C$ such that 
\[
\EE\EQ{\left\| R_{12} - q \right\|^2_n}_u \leq \frac{C}{\sqrt{K_n}}. 
\]
\end{lemma}
\begin{proof}
Writing $R_{12} = [ R_{12}(i) ]_{i=1}^n$, we first show that 
\begin{equation}
\label{R-<R>}
\EE\EQ{(R_{12}(i) - \EQ{R_{12}(i)}_u)^2}_u \leq C / \sqrt{K_n} .
\end{equation}
We write 
\begin{align*} 
\EE\EQ{(R_{12}(i) - \EQ{R_{12}(i)}_u)^2}_u &= 
\EE\EQ{R_{12}(i)^2}_u - \EE\EQ{R_{12}(i)}_u^2 \\ 
 &= \EE\sum_{j,\ell} t^2 s_{ij} s_{i\ell} 
  \EQ{\sigma_j^1\sigma_j^2\sigma_\ell^1 \sigma_\ell^2}_u 
 - \EE\left( \sum_j t s_{ij} \EQ{\sigma_j}_u^2 \right)^2 \\ 
&= \sum_{j\ell} t^2 s_{ij} s_{i\ell} \EE\left[ 
   \EQ{\sigma_j\sigma_\ell}_u^2 - \EQ{\sigma_j}_u^2 \EQ{\sigma_\ell}_u^2 
  \right] , 
\end{align*} 
where $\sigma^1=[\sigma^1_i]$ and $\sigma^2 = [\sigma^2_i]$.  By
Assumption~\ref{ass-K}, the contribution of the terms $j=\ell$ is bounded by
$C/K_n$. Regarding the terms $j\neq \ell$, we note that
$\EQ{\sigma_j\sigma_\ell}_u^2 - \EQ{\sigma_j}_u^2 \EQ{\sigma_\ell}_u^2 =
m_{j\ell}(m_{j\ell} + 2 m_j m_\ell)$.  By Lemma~\ref{lm-var} used in
conjunction with the inequality $\EE m_j m_\ell m_{j\ell} \leq (\EE
m_{j\ell}^2)^{1/2}$, we  obtain~\eqref{R-<R>}, which leads to the inequality 
\[
\EE \EQ{\left\| R_{12} - \EQ{R_{12}}_u \right\|^2_n}_u \leq C/\sqrt{K_n}.
\]
To obtain the result of the lemma, it remains to prove that 
\begin{equation}
\label{<R>-q} 
\EE \left\| \EQ{R_{12}}_u - q \right\|^2_n \leq C/\sqrt{K_n} .
\end{equation} 
The proof of this result is just an adaptation of the proof 
of~\cite[Proposition 1.2]{adh-bre-vSo-yau-21} to our context. The main 
modifications are related with the fact that Equation~\eqref{eq-q} is no more
a scalar equation. 

We write $p = [p_i]_{i=1}^n =\EQ{R_{12}}_u = tS m^2$. Given a set 
$A = \{i_1,\ldots, i_k\}\subset[n]$ and an index $i\in A$, we write for brevity 
\[
\Tanh_{(i_1,\ldots,i_k)}(i) = 
  \Tanh\left( \sum_{l} \sqrt{u} W_{il} m_{(i_1,\ldots,i_k),l}  
  + \sqrt{1-u} \eta_i \right) . 
\]
Recall that $m_{(i_1,\ldots,i_k),l} = 0$ if $l\in A$, and notice that 
$\{ W_{il} \}_{l\not\in A}$ and $\{ m_{(i_1,\ldots,i_k),l} \}_l$ are 
independent, a fact that we shall used repeatidly in the proof without further 
mention. We also write 
\[
p_{(i_1,\ldots,i_k),i} = \sum_l t s_{il} m_{(i_1,\ldots,i_k),l}^2 
 = \begin{bmatrix} tS m_{(i_1,\ldots,i_k)}^2 \end{bmatrix}_i , 
\] 
to be compared with $p_i = \sum_l t s_{il} m_l^2$. 

Our first purpose is to show that 
\begin{equation}
\label{p} 
\left| \EE p_i -  [t S \EE g(u p + (1-u) q)]_i \right| \leq C / \sqrt{K_n}.
\end{equation}

Using Lemma \ref{preTAP}, we have that 
\[
\EE \left| m_i^2 - \Tanh_{(i),i}^2 \right| \leq 
 2 \EE \left| m_i - \Tanh_{(i),i} \right| \leq 
C / \sqrt{K_n} .
\]

Observing that the conditional distribution of the random variable $\sum_{l}
\sqrt{u} W_{il} m_{(i),l} + \sqrt{1-u} \eta_i$ with respect to the
$\sigma$--field $\mcF_{-i}$ generated by 
$\{ W_{k\ell} \, : \, k,l\neq i, k < l, \ \eta_j \, : \, j\neq i \}$ is 
$\cN(0, u p_{(i),i} + (1-u) q_i)$, we have 
\[ 
\EE \Tanh_{(i),i}^2 = 
 \EE\left[ \EE \left[ \Tanh_{(i),i}^2 \ | \ \mcF_{-i} \right]\right] 
 =  \EE g(u p_{(i),i} + (1-u) q_i).  
\]
By Lemma~\ref{preTAP}, we know that 
$\EE ( m_l - m_{(i),l} )^2 \leq C / K_n$ for $l\neq i$. Therefore, 
$\EE | m_l^2 - \EE m_{(i),l}^2 | \leq \EE | m_l^2 -  m_{(i),l}^2 | +
 \EE |  m_{(i),l}^2 - \EE m_{(i),l}^2 | \leq 
2 (\EE (m_l -  m_{(i),l})^2)^{1/2} + 2\EE m_{(i),l}^2 \leq 
C / \sqrt{K_n}$, and then, 
$\EE | p_i - p_{(i),i} | \leq C / \sqrt{K_n}$. Since $g$ is Lipschitz (see
the proof of Lemma~\ref{lm-q}), we obtain that 
$| \EE g\left( u p_i + (1-u) q_i \right) - 
 \EE g\left( u p_{(i),i} + (1-u) q_i \right) | \leq C / \sqrt{K_n}$, and we
deduce from these bounds that 
\[
\left| \EE m_i^2 - \EE g\left( u p_i + (1-u) q_i \right) \right| \leq
  C / \sqrt{K_n}, 
\]
and the bound~\eqref{p} follows. 

Next, we show that 
\begin{equation}
\label{p2} 
\left| \EE p_i^2 - \EE [tS g(u p + (1-u) q)]_i^2 \right| \leq C / \sqrt{K_n}
\end{equation}
along the same principle. For $i\neq j$, we have by Lemma~\ref{preTAP} again 
that
\[
\EE| m_i^2 m_j^2 - \Tanh_{(i)}(i)^2\Tanh_{(j)}(j)^2 | \leq C/\sqrt{K_n}.
\]
We now need to replace $\Tanh_{(i)}(i)$ and $\Tanh_{(j)}(j)$ with 
$\Tanh_{(i,j)}(i)$ and $\Tanh_{(i,j)}(j)$ respectively. We have 
\begin{align*} 
\EE \left( \Tanh_{(i),i} - \Tanh_{(i,j),i} \right)^2 
 &\leq u \EE 
 \left( \sum_{r}  W_{ir} m_{(i),r} - \sum_{r}  W_{ir} m_{(i,j),r}   
 \right)^2  \\
&\leq 2u \sum_{r\neq j}  \EE W_{ir}^2 \EE( m_{(i),r} - m_{(i,j),r} )^2 
  + 2u \EE W_{ij}^2 m_{(i),j}^2 \\
&\leq C / K_n. 
\end{align*}
This implies that 
$\EE | m_i^2 m_j^2 - \EE \Tanh_{(i,j)}(i)^2 \Tanh_{(i,j)}(j)^2 | 
  \leq C/\sqrt{K_n}$. Denoting as $\mcF_{-(i,j)}$ the $\sigma$--field 
generated by $\{ W_{kl} \, : \, k,l \not\in \{i,j\}, \, k < l, 
   \ \eta_r \, : \, r \not\in \{i,j\} \}$, and writing as 
$\mcL(\cdot \ | \ \mcF_{-(i,j)})$ the conditional distribution with respect
to this $\sigma$--field, we have 
\begin{multline*} 
\mcL\left(
 \begin{bmatrix} 
   \sum_r \sqrt{u} W_{ir} m_{(i,j),r}  + \sqrt{1-u} \eta_i \\
   \sum_r \sqrt{u} W_{jr} m_{(i,j),r}  + \sqrt{1-u} \eta_j 
 \end{bmatrix} \ | \ \mcF_{-(i,j)} \right)  \\ 
 = \cN\left( 0 , 
  \begin{bmatrix} 
   u p_{(i,j),i} + (1-u) q_i & 0 \\ 0 & u p_{(i,j),j} + (1-u) q_j 
  \end{bmatrix} 
  \right) , 
\end{multline*} 
which shows that 
\[
 \EE \Tanh_{(i,j)}(i)^2 \Tanh_{(i,j)}(j)^2 =  
 \EE g(u p_{(i,j)}(i) + (1-u) q_i) g(u p_{(i,j)}(j) + (1-u) q_j).
\]
Similarly to above, we also have 
$\EE | p_{(i,j)}(i) - p_i| + \EE | p_{(i,j)}(j) - p_j|\leq C/\sqrt{K_n}$,
thus, 
\[
 \left| \EE \Tanh_{(i,j)}(i)^2 \Tanh_{(i,j)}(j)^2 - 
 \EE g(u p_i + (1-u) q_i) g(u p_j + (1-u) q_j) \right| \leq C/\sqrt{K_n} 
\]
since $g$ is Lipschitz and bounded. Gathering these bounds, we obtain that
\[
\left| \EE m_i^2 m_j^2 - 
 \EE g(u p_i + (1-u) q_i) g(u p_j + (1-u) q_j) \right| \leq C/\sqrt{K_n},  
\]
and since $p_i^2 = \sum_{k,\ell} s_{ik} s_{i\ell} m_k^2 m_\ell^2$, the bound
\eqref{p2} follows. 

From Inequalities~\eqref{p} and \eqref{p2}, we have 
\[
\left\| \EE p - \EE t S g(X_p) \right\|_n \leq C/\sqrt{K_n} 
 \quad\text{and} \quad 
\left| \EE \| p \|^2_n - \EE \| t S g(X_p) \|_n^2 \right| \leq C/\sqrt{K_n} , 
\]
where we wrote $X_p = up + (1-u) q$ for notational simplicity. 

Since the spectral norm $\| S \|$ satisfies $\| S \| \leq \rownorm{S}$, 
it holds by Assumptions~\ref{ass-K} and~\ref{ass-HT} that the function
$x\in\RR_+^n \mapsto tS g(x)$ is Lipschitz for the Euclidean norm with 
a Lipschitz constant $\alpha < 1$ independent of $n$. From what precedes,
we therefore have 
\begin{align*} 
\EE \| p - \EE p \|^2_n &\leq \EE \| p - tS g(\EE X_p ) \|_n^2 \\
&= \EE \| p \|_n^2 +  \|tS g(\EE X_p )  \|_n^2 
  - 2n^{-1} \ps{\EE p}{ tS g(\EE X_p )  } \\ 
&= \EE \| tS g(X_p) \|^2_n + \| tS g(\EE X_p) \|^2_n 
   - 2n^{-1} \ps{\EE tS g(X_p)}{tS g(\EE X_p)} + \cO(1/\sqrt{K_n}) \\ 
 &= \EE \| tSg(X_p) - tSg(\EE X_p) \|^2_n + \cO(1/\sqrt{K_n}) \\
 &\leq \alpha^2 \EE \| p - \EE p \|^2_n + \cO(1/\sqrt{K_n}) 
\end{align*} 
which implies that 
\begin{equation}
\label{p-Ep} 
\EE \| p - \EE p \|^2_n \leq C/\sqrt{K_n} .
\end{equation} 
From this result, we have 
\[
\| \EE p - tS g(\EE X_p) \|_n \leq \| \EE tS g(X_p) - tS g(\EE X_p) \|_n 
  + C/\sqrt{K_n} \leq \EE \| p - \EE p \|_n + C/\sqrt{K_n}
  \leq C/K_n^{1/4}. 
\]
Finally, 
\[
\| \EE p - q \|_n \leq \| tS g(u\EE p+(1-u)q) - tS g(uq + (1-u) q) \|_n 
   + C/K_n^{1/4} \leq \alpha \| \EE p - q \|_n + C/ K_n^{1/4}, 
\]
which leads to 
\[
\| \EE p - q \|_n \leq C/ K_n^{1/4}. 
\]
Together with~\eqref{p-Ep}, we obtain the bound~\eqref{<R>-q}, and the lemma
is proven. 
\end{proof}

\subsection{Proof of Theorem~\ref{th-F}}
Since $g$ is a bounded function, the quadratic form in the expression of 
$\bs F_n$ is bounded. By Lemma~\ref{lm-q}, $\| q\|_\infty$ is bounded, thus, 
the second term in the expression of $\bs F_n$ is bounded, and the 
boundedness of $\bs F_n$ follows. 

To establish Theorem~\ref{th-F}, we use the classical Guerra's approach 
based on the interpolated Hamiltonian $H_u$. Defining the function $\varphi$ 
on $[0,1]$ as 
\[
\varphi(u) = \frac 1n \EE \log Z_u^{(n)} = 
 \frac 1n \EE \log \sum_{\sigma\in\Sigma_n} e^{H_u(\sigma)} , 
\]
we have 
\begin{align*}
\varphi(0) &= \log 2 + \frac 1n \sum_{i=1}^n 
 \EE\log\cosh\left( \sqrt{q_i} \xi + h \right), \quad 
  \text{and} \\ 
\varphi(1) &= F_n.
\end{align*} 
We need to compute the derivative 
$\varphi'(u) = n^{-1} \EQ{\partial_u H_u(\sigma)}_u$. Writing  
\[
U(\sigma^1,\sigma^2) = \EE (\partial_u H_u(\sigma^1)) H_u(\sigma^2) 
 = \frac t4 (\sigma^1\sigma^2)^\T S (\sigma^1\sigma^2) 
  - \frac{t}{2} \ps{(\sigma^1\sigma^2)}{Sg(q)}, 
\]
we know by the well-known Gaussian integration by parts formula, 
see, \emph{e.g.}, \cite[Lemma 1.1]{pan-livre13}, that  
\begin{align*} 
\varphi'(u) &= \frac 1n 
   \EE\EQ{U(\sigma^1,\sigma^1) -  U(\sigma^1,\sigma^2)}_u \\ 
 &= \frac{t}{4n} (1-g(q))^\T S(1-g(q)) - \frac{t}{4n}  
  \EE\EQ{(\sigma^1\sigma^2 - g(q))^\T S (\sigma^1\sigma^2 - g(q))}_u . 
\end{align*} 
Noticing that $\| \sigma^1\sigma^2 - g(q) \|_\infty \leq 2$ and using 
Lemma~\ref{lm-R12}, we have 
\[
 \frac{t}{4n} 
 \left| \EE\EQ{(\sigma^1\sigma^2 - g(q))^\T S (\sigma^1\sigma^2 - g(q))}_u 
  \right| \leq 
\frac 12 \EE\EQ{\| tS (\sigma^1\sigma^2 - g(q)) \|_n }_u 
= \frac 12 \EE\EQ{\| R_{12} - q \|_n }_u \leq 
 \frac{C}{K_n^{1/4}} 
\]
where we recall that the constant $C$ does not depend on $u$. Writing 
\[
F_n = \varphi(1) = \varphi(0) + \int_0^1 \varphi'(u) du , 
\]
and using the last bound, we obtain the result of Theorem~\ref{th-F}. 

\subsection{Proof of Theorem~\ref{th-amp}}

In all the remainder of the paper, we shall work on the original 
Hamiltonian defined in the introduction, and the use of Lemma~\ref{preTAP} 
will be restricted to $u = 1$. 

The first bound obtained for $u=1$ in the statement of Lemma~\ref{preTAP} can 
be rewritten as 
\[
\max_{i\in[n]} \EE(m_i - \Tanh( [ W m_{(i)} ]_{i} ))^2 \leq C / K_n. 
\]  
By a straightforward adaptation of this lemma, we obtain that for a fixed 
integer $M > 0$, it holds that  
\begin{equation}
\label{preTAP(A)} 
\max_{A\subset[n], |A| = M} \max_{i\in\Ac} 
 \EE \left( m_{(A),i} - \Tanh\left( 
 \left[ W m_{(A\cup\{i\})} \right]_{i} \right) 
  \right)^2 \leq \frac{C}{K_n} , 
\end{equation} 
This bound will be at the basis of our proof. Let us fix an integer $k > 0$. 
Given indices $i_1,i_2,\ldots, i_k \in [n]$ which are all different, we obtain
thanks to the bound~\eqref{preTAP(A)} that 
\begin{align} 
\label{mi1} 
m_{i_1} &= \Tanh\left( \left[ W m_{(i_1)} \right]_{i_1} \right) + e_1 \\
m_{(i_1),i_2} &= \Tanh\left( \left[ W m_{(i_1,i_2)} \right]_{i_2} 
               \right) + e_2 \nonumber \\
\cdots \nonumber \\ 
 m_{(i_1,\ldots,i_{k-1}), i_k} &= 
  \Tanh\left( \left[ W m_{(i_1,\ldots,i_{k})} \right]_{i_k} \right) + e_k 
\nonumber 
\end{align}
with $\EE e_l^2 \leq C / K_n$ for each $l\in[k]$. 

This construction will be at the basis of a series of approximations of the
vector $m$ ending with the one provided by the AMP algorithm of 
Theorem~\ref{th-amp}. 
We first define the sequence of families of $\RR^n$--valued vectors
$\{y^1_{(A_1)} \}_{A_1\subset[n], |A_1| = k-1}$, 
$\{y^2_{(A_2)} \}_{A_2\subset[n], |A_2| = k-2}$, ..., $\{y^k \}$ as follows.
We write $y^l_{(A_l)} = [ y^l_{(A_l),i} ]_{i\in[n]}$, and we set 
$y^l_{(A_l),i} = -h$ if $i\in A_l$ in such a way that 
$\Tanh(y^l_{(A_l),i}) = 0$ if $i\in A_l$. With this convention, given indices 
$i_1,i_2,\ldots, i_k \in [n]$ which are all different, we set  
\begin{align*} 
y^1_{(i_1,\ldots,i_{k-1}),i_k} &= 
 \left[ W m_{(i_1,\ldots,i_{k})} \right]_{i_k} 
   \\ 
y^2_{(i_1,\ldots,i_{k-2}),i_{k-1}} &= 
 \left[ W  \Tanh\left( y^1_{(i_1,\ldots,i_{k-1})}  \right) \right]_{i_{k-1}} \\ 
y^3_{(i_1,\ldots,i_{k-3}),i_{k-2}} &= 
 \left[ W  \Tanh\left( y^2_{(i_1,\ldots,i_{k-2})}  \right) \right]_{i_{k-2}} \\ 
 \cdots  \\
y^{k}_{i_1} &= \left[ W 
  \Tanh\left( y^{k-1}_{(i_1)}  \right) \right]_{i_1} . 
\end{align*}
The same kind of construction is provided in~\cite{che-tan-21}. To make things
clearer to the reader, let us assume that $k=3$. Then we have 
\begin{align*}
y^3_{i_1} &= \sum_{i_2\not\in\{i_1\}} W_{i_1i_2} \Tanh\left( 
  \sum_{i_3\not\in\{i_1,i_2\}} W_{i_2i_3} \Tanh\left( 
 \sum_{i_4\not\in\{i_1,i_2,i_3\}} W_{i_3i_4} m_{(i_1,i_2,i_3),i_4}\right)
 \right) .
\end{align*} 
We notice here that the family $\{ W_{i_1i_2} \}_{i_2}$ is independent of the
family of random variables $\Tanh(...)$ that follow these terms in the first
summand, the family $\{W_{i_2i_3}\}_{i_3\not\in\{i_1,i_2\}}$ is independent
from what follows, and so on. More formally, by writing $y^0_{(i_1,\ldots,i_k)}
= \Tanh^{-1}(m_{(i_1,\ldots,i_{k})})$, we have  
\[
y^{l}_{(i_1,\ldots,i_{k-l}),i_{k-l+1}} = 
 \sum_{j\not\in \{i_1,\ldots,i_{k-l+1}\}}  
 W_{i_{k-l+1},j} \Tanh(y^{l-1}_{(i_1,\ldots,i_{k-l+1}),j}), \quad l\in[k], 
\]
and we observe that the families $\{W_{i_{k-l+1},j}\}_{j\not\in
\{i_1,\ldots,i_{k-l+1}\}}$ and 
$\{\Tanh(y^{l-1}_{(i_1,\ldots,i_{k-l+1}),j}\}_{j\not\in
\{i_1,\ldots,i_{k-l+1}\}}$ are independent for each $l\in[k]$. 
The same kind of remark will hold for the next two algorithms. We call
this phenomenon the ``independence along a path of indices''. 

The next algorithm is similar to the previous one except for the fact that the
initial value is $\Tanh(0)$ instead of being $m_{(i_1,\ldots,i_{k})}$. Namely,
we define the $\RR^n$--valued vectors 
$\{\tx^1_{(A_1)} \}_{A_1\subset[n], |A_1| = k-1}$, 
$\{\tx^2_{(A_2)} \}_{A_2\subset[n], |A_2| = k-2}$, ..., $\{\tx^k \}$ 
with $\tx^l_{(A_l)} = [ \tx^l_{(A_l),i} ]_{i\in[n]}$ as follows 
(as above, $\tx^l_{(A_l),i} = -h$ if $i\in A_l$): for indices $i_1,\ldots,
i_k$ which are all different, we set 
\begin{align*} 
\tx^1_{(i_1,\ldots,i_{k-1}),i_k} &= 
 \left[ W_{(i_1,\ldots, i_k)}  \Tanh(0) \right]_{i_k} \\ 
\tx^2_{(i_1,\ldots,i_{k-2}),i_{k-1}} &= 
 \left[ W \Tanh\left( \tx^1_{(i_1,\ldots,i_{k-1})} \right) \right]_{i_{k-1}}\\ 
 \cdots  \\
\tx^{k}_{i_1} &= \left[ W 
 \Tanh\left( \tx^{k-1}_{(i_1)} \right) \right]_{i_1} , 
\end{align*}
where $W_{(i_1,\ldots, i_k)}$ is the matrix $W$ in which the rows 
$i_1,\ldots, i_k$ and the columns $i_1,\ldots, i_k$ are set to zero. 

Our next step consists in replacing the function $\Tanh$ with a polynomial. 
Writing 
\[
f(x) = \sum_{\ell=0}^d \alpha_\ell x^\ell 
\]
as a degree--$d$ polynomial, we define the iterates  
$\cx^1_{(i_1,\ldots,i_{k-1})}$, $\cx^2_{(i_1,\ldots,i_{k-2})}$, ..., $\cx^k$ 
with the same notational conventions as above as 
\begin{align*} 
\cx^1_{(i_1,\ldots,i_{k-1}),i_k} &= 
 \left[ W_{(i_1,\ldots,i_k)} f(0) \right]_{i_k} \\ 
\cx^2_{(i_1,\ldots,i_{k-2}),i_{k-1}} &= 
 \left[ W  f(\cx^1_{(i_1,\ldots,i_{k-1})}) \right]_{i_{k-1}} \\ 
  \cdots  \\
 \cx^{k}_{i_1} &= \left[ W f( \cx^{k-1}_{(i_1)} ) \right]_{i_1} , 
\end{align*} 
by setting $f(\cx^l_{(A_l),i}) = 0$ if $i\in A_l$. 
The iterations for these two last algorithms can be rewritten as follows
for later use. Writing $A_l = \{ i_1, \ldots, i_{k-l} \}$ for $l=0,\ldots, k$, 
we have for $l\in[k]$ 
\begin{align} 
\tx^{l}_{(A_l),i_{k-l+1}} &= \sum_{r\not\in A_{l-1}} W_{i_{k-l+1},r} 
   \Tanh(\tx^{l-1}_{(A_{l-1}),r}), 
 \quad\text{and}  \label{tx} \\ 
\cx^{l}_{(A_l),i_{k-l+1}} &= \sum_{r\not\in A_{l-1}} W_{i_{k-l+1},r} 
   f(\cx^{l-1}_{(A_{l-1}),r})  , 
\label{cx} 
\end{align} 
starting with $\cx^0_{(A_0),r} = \tx^0_{(A_0),r} = 0$. 

As said above, the three preceding algorithms share the property of the
independence along a path of indices.  In order to be able to use the approach
of \cite{bay-lel-mon-15} and \cite{hac-24} as announced at the beginning of
this section, we need to introduce another kind of dependence, namely the one
based on the so-called Non-Backtracking (NB) iterations.  At every iteration
$l$, our next algorithm produces a family of vectors 
$\cz^l_{(j)} = [ \cz^l_{(j),i} ]_{i\in[n]}$ as follows. 
We initialize the algorithm with $\cz^0_{(j),i} = 0$, and write 
\[
\cz^{\ell+1}_{(j),i} = \sum_{r\neq j} W_{ir} f(\cz^\ell_{(i),r}) . 
\] 
Furthermore, stopping at Iteration $k$, we set 
\[
\cz^{k}_{i} = \sum_{r} W_{ir} f(\cz^{k-1}_{(i),r}) . 
\]
Our last intermediate is the following AMP algorithm with the polynomial 
activation function $f$. Starting with $z^0 = 0$ and $z^1 = W f(0)$, 
it reads  
\begin{equation}
\label{amp-z} 
z^{l+1} = W f\left(z^l\right) - \diag\left( (W\odot W) f'(z^l) \right) 
    f\left(z^{l-1}\right) . 
\end{equation} 
For a better readability, we summarize the main features of these five 
algorithms along with the AMP algorithm~\eqref{amp-x} in the following table:
\begin{center} 
\begin{tabular}{|c|c|c|c|}
\hline 
      & Activation function & Algorithm structure  & Initialization \\ \hline   
$y^l_{(\cdots)}$   &  $\Tanh$            & independence along a path of indices & $m_{(\cdots)}$ \\ \hline 
$\tx^l_{(\cdots)}$ &  $\Tanh$            
    & independence along a path of indices & $\Tanh(0)$ \\ \hline 
$\cx^l_{(\cdots)}$ & polynomial $f$      
  & independence along a path of indices & $f(0)$ \\ \hline 
$\cz^l_{(\cdot)}$ & polynomial $f$      & NB         & $f(0)$ \\ \hline 
$z^l$   & polynomial $f$      & AMP         & $z^0 = 0, z^1 = W f(0)$ \\ \hline 
$x^l$   & $\Tanh$             & AMP         
    & $x^0 = 0, x^1 = W \Tanh(0)$ \\ \hline 
\end{tabular} 
\end{center} 

We shall develop below a sequence of approximation results starting with 
the vector $m$ and ending with $x^k$. Before we begin, some new notations and 
preliminary results are necessary. 

We say that a real continuous function $\varphi$ belongs to the set $\PL$ of 
pseudo-Lipschitz functions if there exists $C > 0$ and an integer $a > 0$ such 
that 
\[
\left| \varphi(x) - \varphi(y)  \right| \leq 
C |x - y | \left( 1 + |x|^a + |y|^a\right). 
\]
It is easy to show that each polynomial belongs to $\PL$, and so is the case of
the functions $\Tanh$ and $\Tanh'$. Pseudo-Lipschitz functions are conveniently 
used as test functions in the AMP literature, see, 
\emph{e.g.}, \cite{fen-etal-21}.  

Given $i, j \in[n]$ with $i\neq j$, we also denote as $\mcF_{-i}$ and
$\mcF_{-(i,j)}$ the $\sigma$--fields generated by the random variables $\{
W_{kl} \ : \ k<l \ \text{and} \ k,l \in[n]\setminus\{i\} \}$ and $\{ W_{kl} \ :
\ k<l \ \text{and} \ k,l\in[n]\setminus\{i,j\} \}$ respectively.  The notation
$A_l^{(n)} = A_l$ will always refer to a set of indices $A_l \subset [n]$ such
that $|A_l| = k-l$. 

We begin with an approximation result related with the function $\Tanh$ and 
its derivative: 
\begin{lemma}
\label{approx} 
Let $C > 0$.  For each $e > 0$, there exists a polynomial $p_e$ such 
that $p_e(0) = \Tanh(0)$, 
\[
\max_{\alpha \in [0, C]} 
 \EE \left( p_e(\alpha\xi) - \Tanh(\alpha\xi) \right)^2 \leq e 
\quad \text{and} \quad 
\max_{\alpha \in [0, C]} 
 \EE \left( p'_e(\alpha\xi) - \Tanh'(\alpha\xi) \right)^2 \leq e .
\]
\end{lemma} 
\begin{proof}
Given a small $\delta > 0$, it is known, see \cite[Th.~1]{mha-86} 
or~\cite{dvr-tav-57}, that there exists a polynomial $u$ on $\RR$ such that 
\[
\forall x\in\RR, \ \left| u(x) - \Tanh'(x) \right| \leq 
 \delta \exp(\delta x^2) . 
\]
Defining the polynomial $U$ as 
\[
U(x) = \Tanh(0) + \int_0^x u(s) \, ds 
\]
we obtain that 
$| U(x) - \Tanh(x) | = |\int_0^x (u(s) - \Tanh'(s) ) \, ds | \leq 
 \delta |x| \exp(\delta x^2)$. Therefore, given $\alpha > 0$ not too large, we 
have after a simple derivation that 
\begin{align*}
\EE (u(\alpha\xi) - \Tanh'(\alpha\xi))^2 &\leq 
        \delta^2 / \sqrt{1-\delta \alpha^2}, 
   \quad \text{and} \\
\EE (U(\alpha\xi) - \Tanh(\alpha\xi))^2 &\leq 
          \delta \alpha^2 / (1-\delta \alpha^2)^{3/2}. 
\end{align*} 
By assumption, $0\leq \alpha^2\leq C^2$. Thus, by setting $\delta$ small 
enough, we can take $p_e = U$.  
\end{proof} 

Given a polynomial $f$, we need to introduce the sequence $(\cq^l)_{l\in\NN}$
of $\RR_+^n$--valued vectors defined recursively as  
\begin{equation}
\label{iter-cq}
\cq^0 = 0, \quad \cq^{l+1} = tS \EE f\left(\cX^l\right), 
\end{equation} 
where $\cX^l \sim \cN(0, \diag(\cq^l))$.

\begin{lemma}
\label{qcq} 
The iterations $q^l$ defined in the statement of Theorem~\ref{th-amp} satisfy
$\sup_l \| q^l \|_\infty \leq \log 2$.  Let $e$ be a positive number such that
$\sqrt{e} \leq (1- \log 2)/10$. Let $p_e$ be a polynomial such that $p_e(0) =
\Tanh(0)$ and 
\[
\max_{\alpha \in [0, 1]} 
 \EE \left( p_e(\alpha\xi) - \Tanh(\alpha\xi) \right)^2 
  \leq e ,  
\]
whose existence is guaranteed by Lemma~\ref{approx}.  Consider the iterates
$\cq^l$ provided by equations~\eqref{iter-cq}.  with $f = p_e$. Then, $\sup_l
\| \cq^l \|_\infty \leq 1$, and 
$\sup_l \| \cq^l - q^l \|_\infty \leq 10\sqrt{e}$.  
\end{lemma}
\begin{proof}
Recall the expression of $g$ in~\eqref{def-g}. The bound on $\| q^l \|_\infty$ 
follows from $q^0 = 0$, $0\leq g(q) < 1$ and $\rownorm{tS} < \log 2$. 

We now show that $\| \cq^l \|_\infty \leq 1$ and 
$\| \cq^l - q^l \|_\infty \leq 10\sqrt{e}$ by recurrence on $l$. 
This is trivial for $l=0$ since $\cq^0 = q^0$. Assume that this is true for 
$l$. Since $x^2 - y^2 = 2y(x-y) + (x-y)^2$, we have for each $a\in[0,1]$ that 
\[
| \EE p_e(a\xi)^2 - \EE \Tanh(a\xi)^2 | \leq 2
 \EE | p_e(a\xi) - \Tanh(a\xi) | + \EE (p_e(a\xi) - \Tanh(a\xi))^2 
 \leq 2\sqrt{e} + e. 
\]
Remembering that $g$ is $1$--Lipschitz (see the proof of Lemma~\ref{lm-q}),
we obtain from what precedes that 
\begin{align*} 
\| \cq^{l+1} - q^{l+1} \|_\infty &\leq (\log 2)
 \| \EE p_e(\cX^l)^2 - g(q^l) \|_\infty 
 \leq (\log 2)  \left( \| \cq^l - q^l \|_\infty + 
  \| \EE p_e(\cX^l)^2 - g(\cq^l) \|_\infty \right) \\
&\leq 13 (\log 2) \sqrt{e} \leq 10 \sqrt{e}, 
\end{align*} 
and furthermore, 
$\| \cq^{l+1} \|_\infty \leq \log 2 + 10 \sqrt{e} \leq 1$. 
\end{proof} 
We now turn to the approximation results alluded to above. 
\begin{lemma}
\label{v-m}
It holds that $\EE(\Tanh(y_i^k) - m_i)^2 \leq C / K_n$. 
\end{lemma}
\begin{proof} 
It is clear from the bound~\eqref{preTAP(A)} that $\EE\left(
\Tanh(y^1_{(i_1,\ldots,i_{k-1}),i_k}) - m_{(i_1,\ldots,i_{k-1}),i_k} 
 \right)^2 \leq C/K_n$. Given $l \geq 2$, assume that 
$\EE\left( \Tanh(y^{l}_{(i_1,\ldots,i_{k-l}),i_{k-l+1}}) - 
   m_{(i_1,\ldots,i_{k-l}),i_{k-l+1}} \right)^2  
 \leq C/K_n$. Since $\Tanh$ is $1$-Lipschitz, we have thanks to the 
independence along a path of indices property that 
\begin{multline*} 
 \EE\left( \Tanh(y^{l+1}_{(i_1,\ldots,i_{k-l-1}),i_{k-l}}) - 
  m_{(i_1,\ldots,i_{k-l-1}),i_{k-l}}  \right)^2 \\
= \EE\left( 
 \Tanh\left( \left[ W \Tanh(y^l_{(i_1,\ldots,i_{k-l})})  
   \right]_{i_{k-l}} \right) 
  - \Tanh\left( \left[ W m_{(i_1,\ldots,i_{k-l})}  \right]_{i_{k-l}} \right) 
  - e_l \right)^2  
\end{multline*}
where $\EE e_l^2 \leq C / K_n$. Therefore, 
\begin{align*} 
&\EE\left( \Tanh(y^{l+1}_{(i_1,\ldots,i_{k-l-1}),i_{k-l}}) - 
  m_{(i_1,\ldots,i_{k-l-1}),i_{k-l}}  \right)^2 \\ 
&\leq 2 
 \EE\left( \left[ W \Tanh(y^l_{(i_1,\ldots,i_{k-l})}) \right]_{i_{k-l}} 
  - \left[ W m_{(i_1,\ldots,i_{k-l})}  \right]_{i_{k-l}} \right)^2  
   + C/K_n \\ 
&\leq 2 \sum_r t s_{i_{k-l},r} 
  \EE \left(  
\Tanh(y^{l}_{(i_1,\ldots,i_{k-l}),r}) - m_{(i_1,\ldots,i_{k-l}),r} 
 \right)^2 + C/K_n \\
&\leq C/K_n. 
\end{align*} 
\end{proof}

\begin{lemma}
\label{u-v} 
It holds that $\EE ( \tx^k_i - y^k_i)^2 \leq 4(\log 2)^k$. 
\end{lemma}
\begin{proof}
By recurrence. For different indices $i_1,\ldots, i_k$, we have  
\[
\EE (\tx^1_{(i_1,\ldots,i_{k-1}),i_k} - y^1_{(i_1,\ldots,i_{k-1}),i_k} )^2
 \leq t \sum_{r} s_{i_k,r} \EE ( \Tanh(0) - m_{(i_1,\ldots,i_k),r})^2 
 \leq 4 \log 2.
\]
For $l \geq 1$, assume that 
$\EE (\tx^l_{(i_1,\ldots,i_{k-l}),i_{k-l+1}}- 
  y^l_{(i_1,\ldots,i_{k-l}),i_{k-l+1}} )^2 \leq 4 (\log 2)^l$. Then, by using 
the independence along a path of indices property and doing the same 
calculation as above, we get that 
$\EE (\tx^{l+1}_{(i_1,\ldots,i_{k-l-1}),i_{k-l}}- 
  y^{l+1}_{(i_1,\ldots,i_{k-l-1}),i_{k-l}} )^2 \leq 4 (\log 2)^{l+1}$. 
\end{proof}

\begin{proposition} 
\label{AMP1} 
For each $b > 0$ and each $l \in[k]$, it holds that 
\begin{equation} 
\label{mombnd} 
\sup_n \max_{A_l^{(n)}} \max_{i\in[n]} 
  \EE \left|\cx^{(n),l}_{(A_l^{(n)}),i}\right|^b < \infty .
\end{equation} 
For a function $\varphi\in\PL$, an integer $l\in[k]$, a 
sequence of sets $\cS^{(n)} \in [n]$ with $| \cS^{(n)} | \to\infty$, 
a sequence of $|\cS^{(n)}|$--tuples $(\beta^{(n)}_i)_{i\in\cS^{(n)}}$ such 
that $|\beta^{(n)}_i|\leq 1$, a sequence of sets of the type $A_l^{(n)}$, and 
a number $b > 0$ which are all arbitrary, it holds that 
\begin{equation}
\label{amp-cx} 
\EE\left| \frac{1}{|\cS^{(n)}|} \sum_{i\in\cS^{(n)}} 
 \beta_i^{(n)} \varphi(\cx^{(n),l}_{(A_l^{(n)}),i}) - \beta_i^{(n)} 
  \EE \varphi(\cX^{(n),l}_i) \right|^b
 \xrightarrow[n\to\infty]{} 0. 
\end{equation} 
\end{proposition} 
\begin{proof}
In the expression \eqref{cx} of $\cx^l_{(A_l), i_{k-l+1}}$, we have set 
$|A_l| = k-l$. We need to extend a bit this expression to include a set of 
indices $B \subset[n]$ that might be larger than $A_l$ by writing for 
$i\not\in B$: 
\[
\cx^l_{(B), i} = \sum_{r\not\in B} W_{ir} f(\cx^{l-1}_{(B\cup\{i\}), r}) ,  
\]
which provides consistent iterations one we set 
$\cx^0_{(\cdot),\cdot} = 0$.  We also write 
\[
(\varsigma_{(B),i}^l)^2 = t \sum_{r\not\in B} s_{ir} 
  f(\cx^{l-1}_{(B\cup\{i\}),r})^2 . 
\]
We notice that $(\varsigma_{(B),i}^l)^2$ is the conditional variance of 
$\cx^l_{(B), i}$ given $\mcF_{-i}$, a fact that we shall use repeatidly in
the proof. 

The moment bound \eqref{mombnd} can be proven by recurrence on $l$. For $l=1$,
consider a set $A_1$ and an index $i\not\in A_1$. Since 
$\cx^1_{(A_1),i} = f(0) \sum_{r\not\in A_1} W_{ir}$, it is clear that the
bound~\eqref{mombnd} holds true for $l = 1$. Assuming~\eqref{mombnd} is true 
for $l$, let $i\in A_l$ and $A_{l+1} = A_l \setminus\{i\}$. We have
here $\cx^{l+1}_{(A_{l+1}),i} = \sum_{r\not\in A_{l+1}} W_{ir} 
f(\cx^l_{(A_l),r})$, and thus, 
\[
\EE \left| \cx^{l+1}_{(A_{l+1}),i} \right|^b = 
\EE \EE\left[ \left| \cx^{l+1}_{(A_{l+1}),i} \right|^b \ | \ \mcF_{-i} \right]
= \EE (\varsigma_{(A_{l+1}),i}^{l+1})^b  \EE|\xi|^b 
\]
which is bounded by the recurrence assumption. 

We now prove by recurrence on $l$ the convergence~\eqref{amp-cx} as well as
\begin{equation}
\label{negl}
\EE ( \cx^l_{(A_l),r} - \cx^l_{(A_l\cup\{i\}),r} )^2 \to 0 
\end{equation}
for all sequences $(A_l^{(n)})$, $(i_n)$ with $i_n\not\in A_l^{(n)}$, 
and $(r_n)$ with $r_n \not\in A_l^{(n)}\cup\{i_n\}$.  

We first notice that at the left hand side of~\eqref{amp-cx}, the terms for
which $i\in A_l$ have a negligible contribution. Thus, in all the remainder 
of the proof, we can assume without generality loss that 
$\cS \cap A_l = \emptyset$ for each $l\in [k]$. 

Let us start our recurrence with $l=1$. As above, consider a set $A_1$ and an 
index $i\not\in A_1$. We have 
\[
\EE \varphi(\cx^{1}_{(A_1),i}) = \EE 
 \varphi\left( f(0) \sum_{r\not\in A_1} W_{ir} \right) 
= \EE \varphi\left( \varsigma_{(A_1),i}^{1} \xi  \right) . 
\]

We also have  $|\varphi(x) - \varphi(y) | \leq C|x-y|(1+|x|^a + |y|^a)$ for
some $C,a > 0$.  Recalling that $\cX^l = [\cX^l_i] \sim\cN(0,\diag(\cq^l))$,
we can write that $\cX^1_i = \sqrt{\cq^1_i} \xi$. Note also that
$\cq^1_i = t f(0)^2 \sum_{r} s_{ir}$ and that 
$(\varsigma_{(A_1),i}^1)^2 = t f(0)^2 \sum_{r\not\in A_l} s_{ir}$ with 
$|A_1| = k-1$. With this, we have 
\[
\left| \EE \varphi(\cx^{1}_{(A_1),i}) - \EE\varphi(\cX^1_i) \right| 
 \leq C | \varsigma_{(A_1),i}^1 - \sqrt{\cq^1_i} | \EE |\xi|  
  (1 + (\varsigma_{(A_1),i}^1)^a |\xi|^a + (\cq^1_i)^{a/2} |\xi|^a )  
\]
which converges to zero for each sequence of sets $(A_1^{(n)})$ and sequence
of indices $(i_n)$. We now show that 
\begin{equation} 
\label{cx1}
\frac{1}{|\cS|^2} \EE\left(\sum_{i\in\cS} \beta_i \varphi(\cx^1_{(A_1),i}) 
  - \beta_i \EE \varphi(\cX^1_i) \right)^2 \to 0 . 
\end{equation} 
By expanding the square, we obtain a sum $\sum_{i,j\in \cS}\cdots$. The 
diagonal $i=j$ is easily shown to be $\cO(1/|\cS|)$. Let us show that for 
$i \neq j$, it holds that 
\begin{equation}
\label{varl=1} 
 \EE \left( \varphi(\cx^1_{(A_1),i}) - \EE \varphi(\cX^1_i) \right)
 \left( \varphi(\cx^1_{(A_1),j}) - \EE \varphi(\cX^1_j) \right) \to 0 .
\end{equation} 
Since we showed that 
$\EE\varphi(\cx^1_{(A_1),i}) - \EE \varphi(\cX^1_i) \to 0$, it remains to show 
that $\EE \varphi(\cx^1_{(A_1),i}) \varphi(\cx^1_{(A_1),j}) - 
\EE \varphi(\cX^1_i) \EE \varphi(\cX^1_j) \to 0$ to obtain~\eqref{varl=1}. 
We first have 
\[
\EE (\cx^1_{(A_1), i} - \cx^1_{(A_1\cup\{j\}), i})^2 = 
 f(0)^2 \EE W_{ij}^2 \leq C / K_n 
\]
which establishes in passing~\eqref{negl} for $l=1$. Also, by using the 
pseudo-Lipschitz property of $\varphi$, we have (details for obtaining the 
terms $o_n(1)$ below omitted): 
\begin{align*}
 \EE \varphi(\cx^1_{(A_1),i}) \varphi(\cx^1_{(A_1),j}) 
&= \EE \varphi(\cx^1_{(A_1\cup\{j\}),i}) \varphi(\cx^1_{(A_1\cup\{i\}),j}) 
 + o_n(1)  \\ 
&= \EE 
\varphi\left( f(0) \sum_{r\not\in A_1 \cup \{ j \}} W_{ir} \right)
\varphi\left( f(0) \sum_{r\not\in A_1 \cup \{ i \}} W_{jr} \right) 
                         + o_n(1) \\
&= \EE \varphi(\varsigma^1_{(A_1\cup \{j\}),i} \xi) 
    \EE \varphi(\varsigma^1_{(A_1\cup \{i\}),j} \xi) + o_n(1) \\
&= \EE \varphi(\cX^1_{i}) \EE \varphi(\cX^1_{j}) + o_n(1) , 
\end{align*} 
hence~\eqref{varl=1}. By the moment bound~\eqref{mombnd} and dominated 
convergence, the convergence~\eqref{cx1} holds true. 

Given any moment $b > 2$, we also have 
\[
\sup 
\frac{1}{|\cS|} \left( \EE \left| \sum_{i\in\cS} 
 \beta_i \varphi(\cx^{1}_{(A_1),i}) - \beta_i 
  \EE \varphi(\cX^1_i) \right|^b\right)^{1/b} 
\leq \sup \frac{1}{|\cS|} \sum_{i\in\cS} 
 \left( \EE\left| \varphi(\cx^1_{(A_1),i}) - \EE \varphi(\cX^1_i) 
  \right|^b\right)^{1/b} < \infty 
\]
thanks to~\eqref{mombnd}, where the $\sup$ is taken on $\cS^{(n)}$, 
$(\beta_i^{(n)})_{i\in\cS^{(n)}}$, and $A_1^{(n)}$. The 
convergence~\eqref{amp-cx} follows for $l=1$. 

Assume the recurrence assumption is true for $l$. Letting $i\in A_l$ and
$A_{l+1} = A_l \setminus\{i\}$, we have 
\[
\EE \varphi(\cx^{l+1}_{(A_{l+1}),i}) = \EE \EE\left[ 
 \varphi\left( \sum_{r\not\in A_l} W_{ir} f(\cx^l_{(A_l),r}) \right) 
  \ | \ \mcF_{i} \right] 
= \EE \varphi(\varsigma_{(A_{l+1}),i}^{l+1} \xi )
\]
Recalling that $\cq^{l+1}_i = t\sum_{r} s_{ir} \EE f(\cX^l_r)^2$, we have
by using the recurrence assumption and the fact that $| A_l | = k-l$ is 
fixed that 
\[
\EE\left( (\varsigma_{(A_{l+1}),i}^{l+1})^2 - \cq^{l+1}_i \right)^2 \to 0 . 
\]
Writing $\cX^{l+1}_i = \sqrt{\cq^{l+1}_i} \xi$, we have 
\[
| \EE \varphi(\varsigma_{(A_{l+1}),i}^{l+1} \xi ) - \EE \varphi(\cX^{l+1}_i) | 
\leq C \EE| (\varsigma_{(A_{l+1}),i}^{l+1} - \sqrt{\cq^{l+1}_i} |) |\xi|  
 (1 + (\varsigma_{(A_{l+1}),i}^{l+1})^a |\xi|^a 
   + (\cq^{l+1}_i)^{a/2} |\xi|^a ), 
\]
we then have that 
\[
\EE \varphi(\cx^{l+1}_{(A_{l+1}),i}) - \EE \varphi(\cX^{l+1}_i) \to 0 , 
\]
by Cauchy-Schwarz and the bound~\eqref{mombnd}. 

We now show that 
\begin{equation} 
\label{cxX}
\frac{1}{|\cS|^2} \EE\left(\sum_{i\in\cS} 
 \beta_i \varphi(\cx^{l+1}_{(A_{l+1}),i}) 
  - \beta_i \EE \varphi(\cX^{l+1}_i) \right)^2 \to 0 . 
\end{equation} 
As for the case $l=1$, this will be true if we prove that 
\[
\EE \varphi(\cx^{l+1}_{(A_{l+1}),i}) \varphi(\cx^{l+1}_{(A_{l+1}),j}) - 
\EE \varphi(\cX^{l+1}_{i}) \EE \varphi(\cX^{l+1}_{j}) \to 0
\]
for $i\neq j$.  We write 
\begin{align*} 
\cx^{l+1}_{(A_{l+1}),i} - \cx^{l+1}_{(A_{l+1}\cup\{j\}),i} &= 
      \sum_{r\not\in A_{l+1}\cup\{j\}} W_{ir} 
    \left( f(\cx^l_{(A_l),r}) 
   - f(\cx^l_{(A_l \cup\{j\}),r})\right)
 +  W_{ij} f(\cx^l_{(A_l),j}) \\
 &= \chi_1 + \chi_2 .
\end{align*}
We obviously have $\EE\chi_2^2 \to 0$. Moreover, since the random vectors
$[W_{ir}]_r$ and 
$[ \cx^l_{(A_l),r}, \cx^l_{(A_l\cup \{j\}),r} ]_r$ 
in the expression above are independent, we have 
\[
\EE\chi_1^2 = \EE \EE \left[ \chi_1^2 \ | \ \mcF_{-i}\right] = 
    \sum_{r\not\in A_{l+1}\cup\{j\}} t s_{ir} \EE 
    \left( f(\cx^l_{(A_l),r}) - f(\cx^l_{(A_l \cup\{j\}),r})\right)^2 
\]
which converges to zero by using the recurrence assumption~\eqref{negl}, 
the bound~\eqref{mombnd}, and the pseudo-Lipschitz property of $f$ as a 
polynomial. We thus obtain that 
\[
\EE\left( \cx^{l+1}_{(A_{l+1}),i} - \cx^{l+1}_{(A_{l+1}\cup\{j\}),i} \right)^2 
  \to 0, 
\]
and the convergence~\eqref{negl} is true for $l+1$. 
Using the pseudo-Lipschitz property of $\varphi$, this last result, and 
the bound~\eqref{mombnd}, we also have 
\begin{align*}
\EE \varphi(\cx^{l+1}_{(A_{l+1}),i}) \varphi(\cx^{l+1}_{(A_{l+1}),j}) 
&= \EE \varphi(\cx^{l+1}_{(A_{l+1}\cup\{j\}),i}) 
  \varphi(\cx^{l+1}_{(A_{l+1}\cup\{i\}),j}) + o_n(1) \\
&= \EE \varphi(\varsigma^{l+1}_{(A_{l+1}\cup \{j\}),i} \xi_1) 
    \varphi(\varsigma^{l+1}_{(A_{l+1}\cup \{i\}),j} \xi_2) + o_n(1) , 
\end{align*} 
where $[\xi_1,\xi_2]^\T \sim\cN(0, I_2)$ is a vector independent of everything
else. 
The remainder of the proof is similar to the case $l=1$ with the difference 
that now we need the recurrence assumption to obtain that 
$\EE (\varsigma^{l+1}_{(A_{l+1}\cup \{j\}),i} - \sqrt{\cq^{l+1}_i})^2 \to 0$. 
This leads to~\eqref{cxX}, and this convergence can be upgraded to any 
moment thanks to~\eqref{mombnd}. The proof of Proposition~\ref{AMP1} is 
complete.
\end{proof}

\begin{lemma}
\label{tx-cx}
For each $\varepsilon > 0$, there is a polynomial $f_\varepsilon$ such that 
the iterates~\eqref{cx} with $f = f_\varepsilon$ satisfy 
\[
\limsup_n \max_{i\in[n]} \EE(\cx^k_{i} - \tx^k_{i})^2 \leq \varepsilon.
\]
This polynomial can be chosen in such a way that $f_\varepsilon(0) = 
 \Tanh(0)$, and 
\[
\max_{\alpha \in [0, \sqrt{2}]} 
 \EE \left( f_\varepsilon(\alpha\xi) - \Tanh(\alpha\xi) \right)^2 \leq e 
\] 
for some $e > 0$ that depends only on $\varepsilon$. 
\end{lemma} 
\begin{proof}
Given a small $e > 0$, Lemma~\ref{approx} shows that there exists a 
polynomial $p_e$ such that $p_e(0) = \Tanh(0)$, and
\[
\max_{\alpha \in [0, \sqrt{2}]} 
 \EE \left( p_e(\alpha\xi) - \Tanh(\alpha\xi) \right)^2 \leq e . 
\] 
Unfolding the iterations~\eqref{cx} with $f = p_e$, we shall show by recurrence
on $l$ that for each $l\in[k]$, 
\begin{align*}
& \limsup_n\max_{A_l, i} 
 \EE(\cx^l_{(A_l),i} - \tx^l_{(A_l),i})^2 \leq C e, \ \text{and} \\  
& \limsup_n\max_{A_l, i} 
 \EE \left( p_e(\cx^l_{(A_l),i}) - \Tanh(\tx^l_{(A_l),i}) \right)^2 
  \leq C e ,  
\end{align*} 
where $C > 0$ is a constant that can change from one iteration to another. 
Setting $e = \varepsilon / C$ at the $k^{\text{th}}$ iteration, we obtain our 
result with $f_\varepsilon = p_e$. 

The following results are needed before starting our recurrence. For any 
polynomial $f$, we have from the expressions~\eqref{tx} and~\eqref{cx} that 
for each $A_{l+1}$ and each $i\not\in A_{l+1}$ that 
\[
\mcL\left(\begin{bmatrix} \cx^{l+1}_{(A_{l+1}),i}  \\ \tx^{l+1}_{(A_{l+1}),i}
  \end{bmatrix} \ \mid \ \mcF_i \right) = 
\cN\left(0, R^{l+1}_{(A_{l+1}),i} \right) 
  \quad\text{with}\quad 
 R^{l+1}_{(A_{l+1}),i} = 
  \begin{bmatrix}  
[R^{l+1}_{(A_{l+1}),i}]_{11} & [R^{l+1}_{(A_{l+1}),i}]_{12} \\ 
[R^{l+1}_{(A_{l+1}),i}]_{21} & [R^{l+1}_{(A_{l+1}),i}]_{22} 
 \end{bmatrix} 
\]
satisfying 
\begin{gather*} 
[R^{l+1}_{(A_{l+1}),i}]_{11} = 
t \sum_{r\not\in A_{l}} s_{ir} f(\cx^l_{(A_l),r})^2 , 
\quad 
[R^{l+1}_{(A_{l+1}),i}]_{22} = 
t \sum_{r\not\in A_{l}} s_{ir} \Tanh(\tx^l_{(A_l),r})^2  ,  \ \text{and} \\ 
[R^{l+1}_{(A_{l+1}),i}]_{12} = 
t \sum_{r\not\in A_{l}} s_{ir} f(\cx^l_{(A_l),r}) 
  \Tanh(\tx^l_{(A_l),r})    
\end{gather*} 
with $A_l = A_{l+1} \cup \{ i\}$. By consequence, 
\[
\mcL\left(\cx^{l+1}_{(A_{l+1}),i} - \tx^{l+1}_{(A_{l+1}),i}
   \ \mid \ \mcF_i \right) = 
\cN\left(0, t \sum_{r\not\in A_{l}} s_{ir} 
  \left( f(\cx^l_{(A_l),r}) - \Tanh(\tx^l_{(A_l),r})\right)^2 
  \right) . 
\]
We now tackle our recurrence. For $l=1$, considering a set $A_1$ and an index 
$i\not\in A_1$, we have 
$\cx^1_{(A_1),i} = \tx^1_{(A_1),i} \sim \cN(0, [R^{1}_{(A_1),i}]_{ab})$ 
with 
\[
\forall a,b \in \{1,2\}, \quad 
[R^{1}_{(A_1),i}]_{ab} = t \Tanh(0)^2 \sum_{r\not\in A_1} s_{ir} \leq 1
\]
by Lemma \ref{qcq}. By the construction of $p_e$, we therefore have that 
\[
\EE \left( p_e(\cx^1_{(A_1),i}) - \Tanh(\tx^1_{(A_1),i}) \right)^2 
  \leq e , 
\]
and the recurrence assumption is true for $l = 1$. Assume it is for $l > 1$. 
Let $i\in A_l$ and $A_{l+1} = A_l \setminus \{ i \}$. Then we have 
\begin{align*} 
\EE(\cx^{l+1}_{(A_{l+1}),i} - \tx^{l+1}_{(A_{l+1}),i})^2 
 &= \EE \EE \left[ (\cx^{l+1}_{(A_{l+1}),i} - \tx^{l+1}_{(A_{l+1}),i})^2 \ | \ 
   \mcF_{-i} \right]  \\ 
 &=  t \sum_{r\not\in A_{l+1}} s_{ir} 
  \EE \left( p_e(\cx^l_{(A_l),r}) - \Tanh(\tx^l_{(A_l),r})\right)^2 
\end{align*} 
which is bounded by $C e$ by the recurrence assumption. Since $\Tanh$ is 
Lipschitz, we can furthermore write 
\begin{multline*} 
\EE \left( p_e(\cx^{l+1}_{(A_{l+1}),i}) 
    - \Tanh(\tx^{l+1}_{(A_{l+1}),i}) \right)^2 \leq 
 2 \EE \left( p_e(\cx^{l+1}_{(A_{l+1}),i}) 
    - \Tanh(\cx^{l+1}_{(A_{l+1}),i} ) \right)^2 \\ 
  + 2 \EE(\cx^{l+1}_{(A_{l+1}),i} - \tx^{l+1}_{(A_{l+1}),i})^2,   
\end{multline*} 
and we need to control the first term at the right hand side. Using the 
bound~\eqref{mombnd}, we can write  
\begin{align*} 
& \EE \left( p_e(\cx^{l+1}_{(A_{l+1}),i}) 
   - \Tanh(\cx^{l+1}_{(A_{l+1}),i}) \right)^2  \\ 
&\leq \EE \EE \left[ \left( p_e(\cx^{l+1}_{(A_{l+1}),i}) 
    - \Tanh(\cx^{l+1}_{(A_{l+1}),i}) \right)^2 \ | \ \mcF_{-i} \right]  
    \1_{[R^{l+1}_{(A_{l+1}),i}]_{11} \leq 2 }  
    + C \PP\left[ [R^{l+1}_{(A_{l+1}),i}]_{11} > 2 \right]^{1/2}.
\end{align*} 
The first term at the right hand side is bounded by $e$ by the construction
of $p_e$. Regarding the second term, invoking the previous proposition with 
$\cS = \{ r\in[n] \ : \ s_{ir} > 0, r\not\in A_l \}$, and using that 
$|A_l|$ is fixed, we obtain that 
\[
\max_{A_{l+1}, i} 
 \EE\left( [R^{l+1}_{(A_{l+1}),i}]_{11} - \cq^{l+1}_i \right)^2 \ton 0.
\] 
Since $\cq^{l+1}_i \leq 1$, we obtain that 
$\max_{A_{l+1},i} \PP\left[ [R^{l+1}_{(A_{l+1}),i}]_{11} > 2 \right] 
 \to 0$, and the recurrence assumption is verified for $l+1$. 
\end{proof}

We now manage the iterates $\cx^l$, $\cz^l$ and $z^l$ which are all built
around the polynomial activation function $f(x) = \sum_{\ell=0}^d \alpha_\ell
x^\ell$.  Due to this polynomial nature, we can express these terms with the
help of a tree formalism. The tree structure below  is a simplification of the
structure of \cite{bay-lel-mon-15}. 

Let $T = (V(T), E(T))$ be a rooted tree with the vertex set $V(T)$ and edge set
$E(T)$. The root of this tree is denoted $\circ$, and the distance of a vertex
$u$ to $\circ$ is $|u|$. The root has one child (thus, this is a planted tree),
and every vertex other than $\circ$ can have as much children as the degree $d$
of $f$.  We denote as $\pi(u)$ the parent of the vertex $u$, where the vertices
are oriented towards the root. Thus, $u\to v$ is equivalent to $v = \pi(u)$.
Every vertex $v$ has a label $\ell(v) \in [n]$. The number of children of $v$
is denoted $c(v)$. We denote as $L(T)$ the set of leaves of $T$.  If a leaf $v
\in L(T)$ has a maximal depth in the tree, we attribute to this leaf a number
$c(v) \in \{0,\ldots d\}$ as if this vertex was the parent of $c(v)$ children
which were pruned from the tree. If the depth of this leaf $v\in L(T)$ is not
maximal, then we keep the natural value $c(v) = 0$.  We also use the following
notations: $\overline\cT^k$ is the set of such labelled trees, with depth $k$
at most.  $\widecheck\cT^k \subset\overline\cT^k$ is the subset that satisfies
the following condition: within a tree in $\widecheck\cT^k$, there is no path
$v_1=\circ \leftarrow v_2\leftarrow \cdots \leftarrow v_i$ in which there
exists two identical labels $\ell(\cdot)$. Furthermore,  $\cT^k \subset
\overline\cT^k$ is the subset that satisfies the following non-backtracking
condition: if $v_1=\circ \leftarrow v_2\leftarrow \cdots \leftarrow v_i$ is a
path in a tree belonging to $\cT^k$, then the corresponding sequence of labels
$\ell(\cdot)$ is non-backtracking. This means that for each $j \in [i-2]$, the
three labels $\ell(v_j)$, $\ell(v_{j+1})$ and $\ell(v_{j+2})$ are distinct.
Notice that $\widecheck\cT^k\subset \cT^k$. We also denote as $\cT^k_i$ is the
subset of trees in $\cT^k$ for which the label $\ell(\circ)$ of the root is
$i$, and of course, the label of the child $v$ of the root satisfies $\ell(v)
\neq i$. We apply a similar definition for $\widecheck\cT_i^k$. Notice that
$\widecheck\cT_i^k\subset \cT_i^k$. The trees in $\cT_i^k$ will be
called a non-backtracking trees (NBT). 

Given a tree $T$, write 
\[
W(T) = \prod_{(u\to v) \in E(T)} W_{\ell(u) \ell(v)}, \quad 
\Gamma(T) = \prod_{(u\to v) \in E(T)} \alpha_{c(u)}, \quad \text{and} \quad
\cx(T) = \prod_{v\in L(T)} (\cx^0_{\ell(v)})^{c(v)}.  
\] 

With this formalism, similarly to \cite[Lemma 1]{bay-lel-mon-15}, we have  
\[
\cz_i^k = \sum_{T\in \cT^k_i} W(T) \Gamma(T) \cx(T) 
\quad\text{and}\quad 
\cx_i^k = \sum_{T\in \widecheck\cT^k_i} W(T) \Gamma(T) \cx(T).   
\]
For clarity and completeness, it is perhaps best to illustrate these 
expressions on an example. Considering, \emph{e.g.}, $\cz_i^2$, we have 
\[
\cz^2_i = \sum_{r_2\neq i} W_{i r_2} f(\cz^1_{(i),r_2}) 
 = \sum_{r_2\neq i} W_{i r_2} \sum_{j=0}^d \alpha_j 
 \left( \sum_{r_1\not\in\{i,r_2\}} \alpha_0 W_{r_2r_1} \right)^j . 
\]
Let $r_2 \in [n]\setminus\{i\}$ and $j \in \{ 0,\ldots, d\}$. If $j=0$, then,
the contribution of the indices $r_2$ and $j$ to $\cz^2_i$ is $\alpha_0
W_{ir_2}$. To this term, there corresponds a tree $T\in\cT_i^2$ with depth $1$
such that the child $v$ of the root has a label $\ell(v) = r_2$ and $c(v) = 0$.
For this tree, we clearly have $\alpha_0 W_{ir_2} = W(T) \Gamma(T) \cx(T)$.
The analogous trees with depth $1$ for which $\ell(v) = r_2$ and $c(v) > 0$
have a zero contribution because $\cx(T) = 0$ for these trees. 

If $j > 0$, then, the contribution of the indices $r_2$ and $j$ to $\cz^2_i$ is
a sum of terms of the form $\alpha_0^j \alpha_j W_{ir_2} W_{r_2 r_{1,1}} \ldots
W_{r_2 r_{1,j}}$ where the $r_{2,m}$ scan the set $[n]\setminus\{i,r_2\}$. To
each of these terms, there corresponds a tree $T\in\cT_i^2$ with depth $2$ such
that the child $v$ of the root has a label $\ell(v) = r_2$ and $c(v) = j$, and
such the $j$ leaves $v_m$ have labels $\ell(v_m) = r_{1,m}$ and $c(v_m) = 0$.
We see that we have $W(T) \Gamma(T) \cx(T) = \alpha_0^j \alpha_j W_{ir_2}
W_{r_2 r_{1,1}} \ldots W_{r_2 r_{1,j}}$ for this tree.  The trees in $\cT_i^2$
with depth $2$ and for which there is a leaf $w$ with $c(w) > 0$ have a zero
contribution because $\cx(T) = 0$ for these trees. The equation $\cz^2_i =
\sum_{T\in \cT^2_i} W(T) \Gamma(T) \cx(T)$ follows.

Given an integer $b > 0$ and $b$ trees $T_1, \ldots, T_b \in \overline\cT_i^k$,
define the graph $G = \bs G(T_1,\ldots, T_b)$ as being the rooted, undirected,
and labelled graph obtained by merging the nodes of these trees that have the
same label $\ell(\cdot)$.  This common label will be the label of the resulting
node in $G$. Note that it is possible that two nodes in a same non-backtracking
tree merge.  Of course, the root node $\circ$ of $G$ will have the label
$\ell(\circ) = i$.  The other nodes are numbered, say, in the increasing order
of their labels.  The edges of $G$ are furthermore unweighted. 

For a tree $T$ labelled as above and for $j,l \in [n]$, define 
\[
\phi(T)_{j l} = \left| \left\{ (u\to v) \in E(T), \ \{ \ell(u),\ell(v) \} = 
  \{j, l\} \right\} \right| .  
\]
Finally, recalling Assumption~\ref{ass-card}, we define the set $\cK$ as 
\[
\cK = \left\{ \{i,j\} \subset [n], \ s_{ij} > 0 \right\}, 
\]
and the section $\cK_i$ for $i\in[n]$ as 
\[
\cK_i = \left\{ j \in [n], \ s_{ij} > 0 \right\}. 
\]

The following lemma is proven in Appendix~\ref{prf-cA}. 
\begin{lemma}
\label{cA} 
Let $b, r \geq 2$ be two integers. 
Let $\cA \subset (\cT_i^k)^{\otimes b}$ be such that each $b$--tuple 
$(T_1, \ldots, T_b) \in \cA$ satisfies the two following conditions: 
\begin{itemize}
\item  $|V(\bs G(T_1,\ldots, T_b))| \leq r$.
\item When $\sum_{k=1}^b \phi(T_k)_{jl} > 0$, it holds that $\{j,l\} \in \cK$. 
\end{itemize}
Then, 
\[
\left| \cA \right| \leq C K_n^{r-1}.
\]
\end{lemma}

\begin{proposition}
\label{indep-amp} 
For each even integer $b\geq 2$, it holds that 
$\EE\left( \cx^k_i - \cz^k_i \right)^b \leq C / K_n$ and 
$\EE\left( \cz^k_i - z^k_i \right)^b \leq C / K_n$ . 
\end{proposition}
Combining the results of this proposition with the bound~\eqref{mombnd}, 
we obtain that 
\begin{equation}
\label{momz} 
\forall b > 0, \ 
 \sup_n \max_{i\in[n]} \EE \left| z^l_i \right|^b < \infty
\end{equation}
for each $b > 0$, a bound that will be useful later. 
\begin{proof}
We know that for each $T_1, \ldots, T_b\in\cT^k_i\setminus \widecheck\cT^k_i$, 
it holds that $\sum_{j<l} \phi(T_1)_{jl} +\cdots+ \phi(T_b)_{jl} \leq C_b$ with
$C_b = b(1+d+\cdots+d^{k-1})$. For $r \in [C_b]$, define the set $\cC_i(r)$ as 
\begin{align*}
\cC_i(r) = \Bigl\{ &  (T_1,\ldots,T_b) \, : \, T_1, \ldots, T_b  \in 
     \cT^k_i\setminus \widecheck\cT^k_i, \\ 
 &\forall j < l, \ \phi(T_1)_{jl} + \cdots + \phi(T_b)_{jl} \neq 1 , \\ 
 &\forall j < l, \ \phi(T_1)_{jl} + \cdots + \phi(T_b)_{jl} > 0 \ 
      \Rightarrow \{j,l\} \in \cK,  \\ 
 & \sum_{j<l} \phi(T_1)_{jl} + \cdots + \phi(T_b)_{jl} = r \Bigr\}. 
\end{align*} 
With this definition, we have 
\begin{align*} 
\EE (\cz_i^k - \cx_i^k)^b &=  
 \EE \sum_{T_1, \ldots, T_b \in \cT^k_i\setminus \widecheck\cT^k_i  } 
  \Gamma(T_1) \ldots \Gamma(T_b) x(T_1) \ldots x(T_b) W(T_1) \ldots W(T_b) \\ 
&= \EE \sum_{r=2}^{C_b} \sum_{(T_1,\ldots, T_b) \in \cC_i(r)} 
  \Gamma(T_1) \ldots \Gamma(T_b) x(T_1) \ldots x(T_b) W(T_1) \ldots W(T_b) \\ 
&\leq C \sum_{r=2}^{C_b} \sum_{(T_1,\ldots, T_b) \in \cC_i(r)} 
  \left| \EE W(T_1) \ldots W(T_b) \right|, 
\end{align*} 
where the constant $C$ depends in general on the polynomial $f$. 

Fix $r$ and assume that $\cC_i(r) \neq \emptyset$. For each 
$(T_1,\ldots, T_b) \in \cC_i(r)$, we have  
\[
\left| \EE W(T_1) \ldots W(T_b) \right| = 
 \prod_{j<l} \Bigl| \EE W_{jl}^{\phi(T_1)_{jl}+\cdots+\phi(T_b)_{jl}} \Bigr| 
 \leq C K_n^{- r / 2}. 
\]
We need to show that $| \cC_i(r) | \leq C K_n^{r/2-1}$ to obtain the first
bound in the statement. 
Given $T_1,\ldots, T_b \in \cC_i(r)$, the graph $G = \bs G(T_1,\ldots, T_b)$ 
satisfies 
\[
 | E(G) | = \sum_{j<l} \1_{\phi(T_1)_{jl} +\cdots+ \phi(T_b)_{jl} \geq 2}, 
\]
which shows that $| E(G) | \leq r/2$. For any connected graph $G'$, it is 
well-known that $|V(G')| \leq |E(G')| +1$ with equality if and only if 
$G'$ is a tree. The crucial observation here is that since 
$T_1,\ldots,T_b \in \cC_i(r)$, the graph $G$ is not a tree because there is at
least one label that is repeated in some path belonging to $T_1$ (and similarly
to $T_2, ..., T_b$). Therefore $| V(G) | \leq r / 2$. It remains to apply 
Lemma~\ref{cA} to obtain the first bound.  

We now turn to the second bound. In a directed and labelled graph,  
\begin{itemize}
\item
A backtracking path of length $3$ is a path  
$a \to b \to c \to d$ such that $\ell(a) = \ell(c)$ and $\ell(b) = \ell(d)$. 
\item A backtracking star is a structure 
$a,b \to c \to d$ where $\ell(a) = \ell(b) = \ell(d)$. 
\end{itemize}

In \cite[Lemma 3]{bay-lel-mon-15}, it is shown that  
\[
z_i^k = \cz_i^k + \sum_{T\in \cB_i^k} W(T) \widetilde\Gamma(T) \cz(T) 
\]
where $\widetilde\Gamma(T)$ is bounded and where $\cB_i^k$ is 
a certain subset of $\overline\cT^k_i$ such that each $T \in \cB^k_i$ contains
at least one backtracking path of length $3$ or a backtracking star. 

With this at hand, we have 
\[
\EE(z_i^k - \cz_i^k)^b \leq C 
 \sum_{T_1, \ldots, T_b \in \cB_i^k} \left| \EE W(T_1) \ldots W(T_b) \right| .
\]
Given an integer $r\in [C_b]$, we define the set $\cD_i(r)$ similarly
to $\cC_i(r)$ above except for the fact that $T_1, \ldots, T_b \in \cB_i^k$ 
instead of $\cT^k_i\setminus \widecheck\cT^k_i$. Fixing $r$ such that 
$\cD_i(r) \neq \emptyset$, we have that $| \EE W(T_1) \ldots W(T_b) | 
 \leq C K_n^{-r/2}$ when $(T_1, \ldots, T_b) \in \cD_i(r)$. 
Furthermore, for $G = \bs G(T_1, \ldots, T_b)$, we recall that 
$| E(G) | = \sum_{j<l} \1_{\phi(T_1)_{jl} + \cdots + \phi(T_b)_{jl} \geq 2}$. 
Due to the presence of a backtracking path of length $3$ or a backtracking
star in each of the the trees $T_1, \ldots, T_b$, we observe that 
\[
2 | E(G) | + 2 \leq \sum_{j<l} \phi(T_1)_{jl} +\cdots+ \phi(T_b)_{jl} = r. 
\]
Since $G$ is connected, $|V(G) | \leq | E(G) | + 1  \leq r/2$. 
It remains to apply Lemma~\ref{cA} again. 
\end{proof}

Our last approximation result relates the iterates $z^k$ with the $x^k$: 
\begin{lemma}
\label{z-x} 
Let $C_W > 0$ be a constant, and define the probability event 
$\cE = [ \| W \| \leq C_W ]$. 
For each $\varepsilon > 0$, there is a polynomial $f_\varepsilon$ such that 
the iterates~$z^l$ obtained with $f = f_\varepsilon$ satisfy 
\[
\limsup_n \EE \left[ \left\| z^k - x^k \right\|_n^2 \1_{\cE} \right] 
  \leq \varepsilon . 
\] 
This polynomial can be chosen in such a way that $f_\varepsilon(0) = 
 \Tanh(0)$ and 
\[
\max_{\alpha \in [0, \sqrt{2}]} 
 \EE \left( f_\varepsilon(\alpha\xi) - \Tanh(\alpha\xi) \right)^2 \leq c 
 \quad\text{and} \quad 
\max_{\alpha \in [0, \sqrt{2}]} 
 \EE \left( f_\varepsilon'(\alpha\xi) - \Tanh'(\alpha\xi) \right)^2 \leq c . 
\]
for some $c > 0$ that depends on $\varepsilon$ only. 
\end{lemma} 

The proof is close to \cite[end of the proof of Theorem 2]{hac-24}. 
\begin{proof}
Given a small $c > 0$, Lemma~\ref{approx} shows that there exists a polynomial 
$p_c$ such that $p_c(0) = \Tanh(0)$,  
\[
\max_{\alpha \in [0, \sqrt{2}]} 
 \EE \left( p_c(\alpha\xi) - \Tanh(\alpha\xi) \right)^2 \leq c , 
 \quad\text{and} \quad 
\max_{\alpha \in [0, \sqrt{2}]} 
 \EE \left( p'_c(\alpha\xi) - \Tanh'(\alpha\xi) \right)^2 \leq c . 
\] 
In the proof, $\delta(c)$ will denote a generic function defined near zero in
$\RR_+$ such that $\delta(c) \to 0$ when $c \to 0$.  Constructing the $z^l$'s
with $f = p_c$, we shall show by recurrence on $l$ that for each $l\in[k]$, 
\[
\limsup_n \EE \left\| z^l - x^l \right\|_n^2 \1_{\cE} \leq \delta(c) 
 \quad\text{and}\quad 
\limsup_n \EE \left\| p_c(z^l) - \Tanh(x^l) \right\|_n^2 \1_{\cE} 
  \leq \delta(c) , 
\]
where the function $\delta$ can change from an iteration to another. 
At Iteration $k$, it will be enough to choose $c$ such that $\delta(c) \leq 
\varepsilon$ and to set $f_\varepsilon = p_c$ to obtain the result of the 
lemma. 

Starting with $l=1$, we have $z^1 = x^1 = W f(0) = W \Tanh(0)$. Similarly to 
the beginning of the proof of Lemma~\ref{tx-cx}, we also have 
$\EE \left\| p_c(z^1) - \Tanh(x^1) \right\|_n^2 \leq c$. 

Assume now that the recurrence assumption is true for $l$. Recall the 
expressions \eqref{amp-z} and~\eqref{amp-x} of the $z^l$'s and the $x^l$'s
respectively. Our first task is to show that 
\begin{equation}
\label{ons} 
\limsup_n\EE\left\| \diag ((W\odot W) p_c'(z^l)) p_c(z^{l-1}) - 
  \diag ( tS \EE \Tanh'(X^l) ) \Tanh(x^{l-1}) \right\|_n^2 \1_{\cE} 
 \leq \delta(c). 
\end{equation} 
By making use of the bound~\eqref{momz} along with Cauchy-Schwarz, we get 
\begin{align*} 
\EE\left\| \diag ((W\odot W  - tS )p_c'(z^l)) p_c(z^{l-1}) \right\|_n^2 
 &=  \frac 1n \sum_i 
 \EE \left[ (W\odot W  - tS )p_c'(z^l) \right]_i^2 p_c(z^{l-1}_i)^2 \\  
 &\leq \frac Cn \sum_i \left( \EE \left[ (W\odot W  - tS )p_c'(z^l) \right]_i^4 
   \right)^{1/2} 
\end{align*} 
which converges to zero by \cite[Lemma 21]{hac-24}. We now write 
\begin{align*} 
& \left\| \diag (tS p_c'(z^l)) p_c(z^{l-1}) - 
  \diag ( tS \EE \Tanh'(X^l) ) \Tanh(x^{l-1}) \right\|_n \\
&\leq
\left\| \diag (tS (p_c'(z^l) - \EE p_c'(\cX^l)) p_c(z^{l-1}) 
 \right\|_n 
+ \left\| \diag (tS \EE p_c'(\cX^l)) ( p_c(z^{l-1}) - \Tanh(x^{l-1}) 
  \right\|_n \\
&\phantom{=} 
+ \left\| \diag (tS (\EE p_c'(\cX^l) - \EE\Tanh'(X^l)) \Tanh(x^{l-1}) 
  \right\|_n . 
\end{align*} 
By Cauchy-Schwarz and the bound \eqref{momz}, we have 
\[
\EE \left\| \diag (tS (p_c'(z^l) - \EE p_c'(\cX^l)) p_c(z^{l-1}) 
 \right\|_n^2 \leq 
\frac Cn \sum_i \left( \EE\left( \sum_r s_{ir} 
   \left( p_c'(z^l_r) - \EE p_c'(\cX^l_r) \right) \right)^4 \right)^{1/2} 
\]
which converges to zero by Propositions~\ref{AMP1} and~\ref{indep-amp}. 

We also have that 
\[
| \EE p_c'(\cX^l) | \leq 
| \EE p_c'(\cX^l)  - \EE \Tanh'(\cX_l) | + | \EE \Tanh'(\cX_l) | 
 \leq 1 + \sqrt{c}, 
\]
therefore, $\| \diag (tS \EE p_c'(\cX^l)) \|$ is bounded, and thus, 
\[
\EE \left\| \diag (tS \EE p_c'(\cX^l)) ( p_c(z^{l-1}) - \Tanh(x^{l-1}) ) 
  \right\|_n^2 \1_{\cE} \leq 
C \EE \left\| p_c(z^{l-1}) - \Tanh(x^{l-1}) \right\|_n^2\1_{\cE}  
\]
which $\limsup$ is bounded by $\delta(c)$ by the recurrence assumption. 

We finally have that 
\[
| \EE p_c'(\cX^l_i) - \EE\Tanh'(X^l_i) | \leq 
| \EE p_c'(\cX^l_i) - \EE \Tanh'(\cX^l_i) | + 
  | \EE \Tanh'(\cX^l_i) - \EE\Tanh'(X^l_i) |  \leq \delta(c) 
\]
by the construction of $p_c'$ and by Lemma~\ref{qcq} that shows that 
$\|\cq^l - q^l\|_\infty$ is small. With this, we obtain that 
\[
\EE \left\| \diag (tS (\EE p_c'(\cX^l) - \EE\Tanh'(X^l)) \Tanh(x^{l-1}) 
  \right\|_n^2 \leq \delta(c),  
\]
and \eqref{ons} follows. 

With this, we have 
\begin{multline*} 
\| z^{l+1} - x^{l+1} \|_n \1_{\cE} \leq C_W \| z^l - x^l \|_n \1_{\cE} \\
  + \left\| \diag ((W\odot W) p_c'(z^l)) p_c(z^{l-1}) - 
  \diag ( tS \EE \Tanh'(X^l) ) \Tanh(x^{l-1}) \right\|_n \1_{\cE}
\end{multline*} 
and we get from the recurrence assumption and from~\eqref{ons} that 
$\EE \| z^{l+1} - x^{l+1} \|_n^2 \1_{\cE} \leq \delta(c)$. Also, writing  
\begin{align*} 
\| p_c(z^{l+1}) - \Tanh(x^{l+1}) \|_n \1_{\cE} &\leq 
\| p_c(z^{l+1}) - \Tanh(z^{l+1}) \|_n + \| z^{l+1} - x^{l+1} \|_n \1_{\cE} \\ 
&\leq 
\| p_c(z^{l+1}) - \EE p_c(\cX^{l+1}) \|_n + 
\| \Tanh(z^{l+1}) - \EE \Tanh(\cX^{l+1}) \|_n \\
 &\phantom{=} + \| \EE p_c(\cX^{l+1}) - \EE \Tanh(\cX^{l+1}) \|_n + 
\| z^{l+1} - x^{l+1} \|_n \1_{\cE}, 
\end{align*} 
we see that $\limsup_n \EE \| p_c(z^{l+1}) - \EE p_c(\cX^{l+1}) \|_n^2 \1_\cE 
\leq \delta(c)$, and the recurrence assumption is verified for $l+1$. 
\end{proof} 

\subsubsection*{Theorem~\ref{th-amp}: end of proof} 
Given $\varepsilon > 0$, we know that we can choose $c > 0$ small enough 
in the statement of Lemma~\ref{z-x} so that the conclusions of this lemma
are satisfied. If we make $c$ smaller if necessary, then the conclusion of
Lemma~\ref{tx-cx} will be true for this same polynomial $f_\varepsilon$. 
In this situation, combining Lemma~\ref{z-x}, Proposition~\ref{indep-amp}
with $b=2$, and Lemma~\ref{tx-cx}, we obtain that 
\[
\limsup_n \EE \left\| \tx^k - x^k \right\|_n^2 \1_\cE \leq 2\varepsilon , 
\]
which implies that 
\[
\limsup_n \EE \left\| \Tanh(\tx^k) - \Tanh(x^k) \right\|_n^2 \1_\cE 
    \leq 2\varepsilon . 
\]
Using Lemmas~\ref{u-v} and~\ref{v-m}, we obtain that 
\[
\limsup_n \EE \left\| m - \Tanh(x^k) \right\|_n^2 \1_\cE 
    \leq 2\varepsilon + 4(\log 2)^k . 
\]
Denoting as $\| M \|_\infty$ the max norm of the matrix $M$, we 
know from \cite[Th.~1.1]{ban-vha-16} that
\[
\EE\| W^{(n)} \| \leq T^{(n)}, 
\]
where
\[
T^{(n)} = (1+\delta) 
  \left( 2 \rownorm{S^{(n)}}^{1/2} 
  + \frac{6}{\sqrt{\log(1+\delta)}} (\| S^{(n)} \|_\infty \log n)^{1/2}
 \right) 
\]
for an arbitrary $\delta > 0$. Furthermore, by Gaussian concentration,
\[
\PP\left[ \| W^{(n)} \| \geq T^{(n)} + t \right] 
  \leq \exp(- t^2 / (2 \| S^{(n)}\|_\infty)^2 ) , 
\] 
for $\delta \in (0, 1/2]$, as given by \cite[Cor.~3.9]{ban-vha-16}.  By 
consequence, since $K_n\geq \log n$, it holds that there exists $C_W > 0$ for
which $\1_\cE \to_n 1$ almost surely. This implies that 
\[
\limsup_n \EE \left\| m - \Tanh(x^k) \right\|_n^2 
    \leq 2\varepsilon + 4(\log 2)^k ,  
\]
thus, $\limsup_k \limsup_n \EE \left\| m - \Tanh(x^k) \right\|_n^2 
    \leq 2\varepsilon$. Since $\varepsilon$ is arbitrary, Theorem~\ref{th-amp}
holds true. 

We close the paper with some remarks. 

\begin{remark}
\label{rem-univ} 
The proof for the bound~\eqref{preTAP(A)} requires the Gaussian assumption on
the entries of the matrix $W$. This assumption is not essential for the rest of
the proof which is based on the approach of \cite{bay-lel-mon-15, hac-24}.  If
we manage to generalize the bound~\eqref{preTAP(A)} to the non-necessarily
Gaussian case, the proof above will continue to work after some easy
adaptations.  
\end{remark} 

\begin{remark}
The condition $K_n\geq \log n$ is an artifact of our proof due to the 
fact that we needed to bound the spectral norm $\| W \|$ in order to approximate
our AMP algorithm with a polynomial activation function with the AMP 
algorithm with the $\Tanh$ activation function. Another possible proof 
technique would be to pass from the iterates $\tx^l$ to the iterates $x^l$ 
without the need of introducing the polynomial intermediates. This is left for 
future research. 
\end{remark} 

 \subsection*{Acknowledgement} 
I would like to thank Christian Brennecke for a useful discussion. 

\appendix 

\section{Proof of Lemma~\ref{lm-q}} 
\label{prf-lm-q} 

By a derivative calculation and the use of the Gaussian integration by parts
formula (see, \emph{e.g.} the proof of \cite[Proposition
1.3.8]{tal-livre11-t1}), we can show that $g$ is $1$--Lipschitz.  By
Assumptions~\ref{ass-K} and~\ref{ass-HT}, it holds that 
$\sup_n \rownorm{t S^{(n)}} < 1$. Thus, for each $n > 0$, the function
$q\in\RR_+^n \mapsto t S^{(n)} g(q)$ is a contraction for the
$\|\cdot\|_\infty$ norm on $\RR_+^n$, and the result follows by Banach's fixed
point theorem. Since $0\leq g(q) < 1$, we also have $\| q^{(n)} \|_\infty \leq
\rownorm{t S^{(n)}} \| g(q^{(n)})\|_\infty < \log 2$ by
Assumption~\ref{ass-HT}. 

\section{Proof of Corollary~\ref{cor-band}} 
\label{prf-cor-band} 
We take out the superscripts $^{(n)}$. 
Recall that $S = [ \psi(|i-j|) ]$. For $n$ large enough, let $\Scirc$ be the 
circulant deformation of $S$ given as the $n\times n$ circulant symmetric 
matrix which first row is 
\[
\begin{bmatrix} \psi(0)&\cdots&\psi(K_n)& 0 &\cdots &0& \psi(K_n) &\cdots&
 \psi(1) \end{bmatrix}. 
\]
Given a $n\times n$ matrix $X$ taken from the Gaussian Orthogonal Ensemble, 
let $W = (t S)^{\odot 1/2} \odot X$, and
$\widetilde W = [ \widetilde W_{ij} ] = (t \Scirc)^{\odot 1/2} \odot X$ where
$[A_{ij}]^{\odot 1/2} = [A_{ij}^{1/2}]$. Recall that $H(\sigma) = \sigma^\T W
\sigma /2 + h\ps{\sigma}{1}$ with the free energy $F_n$. Define the Hamiltonian
$\widetilde H(\sigma) = \sigma^\T \widetilde W \sigma / 2 + h\ps{\sigma}{1}$,
and let $\widetilde F_n$ be the associated free energy. Knowing from
Corollary~\ref{cor-sto} that the convergence~\eqref{bisto} holds true for
$\widetilde F_n$, all we need to prove is that $F_n - \widetilde F_n \to 0$. 
We can write $\widetilde H(\sigma) = H(\sigma) + E(\sigma)$ where 
$E(\sigma)$ is given as 
\[
E(\sigma) = \sum_{i=1}^{K_n} \sum_{j=n-K_n+i}^n \sigma_i\sigma_j 
  \widetilde W_{ij}.  
\]
Note that $E(\sigma)$ and $H(\sigma)$ are independent and that the
$\widetilde W_{ij}$'s in the expression of $E(\sigma)$ above satisfy 
$\EE \widetilde W_{ij}^2 = t \psi(|i-j+n|)$. 
By Jensen's inequality, 
$\EE\log(\sum e^{H(\sigma) + E(\sigma)}/\sum e^{H(\sigma)}) 
= \EE\log\EQ{e^{E(\sigma)}} \geq
 \EE\EQ{E(\sigma)} = 0$, thus, $\widetilde F_n \geq F_n$. 
We also have that 
\[
\EE E(\sigma)^2 = t \sum_{i=1}^{K_n} \sum_{j=n-K_n+i}^n 
 \psi(|i-j+n|) \leq t K_n .
\]
Jensen's inequality applied to the expectation with respect to the 
law of $E(\sigma)$ leads to 
\[
\widetilde F_n \leq \frac 1n \EE\log \sum_\sigma 
  e^{H(\sigma) + \EE E(\sigma)^2/2} \leq \frac{t K_n}{2n} + F_n , 
\]
and the corollary is established. 

\section{Proof of Lemma~\ref{lm-var}} 
\label{prf-lm-var} 
Using Identity \eqref{mij-mi}, it is enough to bound 
$\EE (\delta_i m_j^{[i]})^2$. In the derivations below, we use Itô's lemma to 
obtain the equality. To obtain the first inequality, we extract the terms 
$k=j$ from the two sums at the right hand side of the equality, we observe
that $m_{jj}^{[i]} \in[0,1]$, and we use the inequality $ab \leq (a^2+b^2)/2$. 
We obtain 
\begin{align*} 
&\EE (\delta_i m_j^{[i]})^2 \\
&= -2 \sum_{k\neq i} s_{ik} \int_0^{tu} 
 \EE \delta_i m_j^{[i]}(v) \delta_i \left( m_{k}^{{[i]}} m_{kj}^{{[i]}}\right)(v) dv 
 +  \sum_{k\neq i} s_{ik} \int_0^{tu} 
  \EE (\varepsilon_i m_{kj}^{{[i]}}(v))^2 dv \\
&\leq  3 s_{ij} t  
 + \sum_{k\neq i,j} s_{ik} \int_0^t \left( \EE(\varepsilon_i m_{kj}^{{[i]}}(v))^2 
+ \EE(\delta_i \left( m_{k}^{{[i]}} m_{kj}^{{[i]}}\right)(v))^2 \right) dv 
+ \left(\sum_{k\neq i,j} s_{ik} \right) \int_0^t \EE(\delta_i m_j^{[i]}(v))^2 dv \\ 
&\leq  3 s_{ij} t  
 +  \bC_{\text{row}} \max_{k \neq i,j} 
\int_0^t \left( \EE(\varepsilon_i m_{kj}^{{[i]}}(v))^2 
+ \EE(\delta_i \left( m_{k}^{{[i]}} m_{kj}^{{[i]}}\right)(v))^2 \right) dv 
+ \bC_{\text{row}}  \int_0^t \EE(\delta_i m_j^{[i]}(v))^2 dv \\ 
&\leq  3 s_{ij} t  
 +  \bC_{\text{row}} \max_{\substack{k \neq i,j \\ 
    \sigma_i = \pm 1}} 
 \int_0^t \EE (1 + m_{k}^{[i]}(v)^2) m_{kj}^{{[i]}}(v)^2 dv 
 + \bC_{\text{row}} \int_0^t \EE(\delta_i m_j^{[i]}(v))^2 dv . 
\end{align*} 
In the last inequality, we used that 
\begin{multline*} 
\left(\frac{x(1) + x(-1)}{2}\right)^2 + 
\left(\frac{y(1)x(1) - y(-1)x(-1)}{2}\right)^2  \leq 
\frac{x(1)^2 + x(-1)^2 + y(1)^2x(1)^2 + y(-1)^2x(-1)^2}{2} \\ 
\leq \max_{i=\pm 1} x(i)^2 (1 + y(i)^2) .
\end{multline*} 
We now use that $m_{kj}^{[i]} = (1 - (m_k^{[i]})^2) \delta_k m_j^{[i,k]}$ 
to obtain that 
\[
\EE (\delta_i m_j^{[i]})^2 
\leq  3 s_{ij} t  
 +  \bC_{\text{row}} \max_{\substack{k \neq i,j \\ 
    \sigma_i = \pm 1}} 
 \int_0^t \EE(\delta_k m_{j}^{[i,k]}(v))^2 dv 
 + \bC_{\text{row}} \int_0^t \EE(\delta_i m_j^{[i]}(v))^2 dv
\]
Using Gronwäll's inequality which states that if $\varphi(t) \leq \alpha(t) + 
C \int_0^t \varphi(v) dv$, then  
$\varphi(t) \leq \alpha(t) + C \int_0^t \alpha(v) 
 \exp(C(t-v)) dv$, we obtain by making an Integration by Parts that 
\[ 
\EE (\delta_i m_j^{[i]})^2 
\leq  6 s_{ij} t + \bC_{\text{row}} \max_{\substack{k_1 \neq i,j \\ 
    \sigma_i = \pm 1}} 
 \int_0^t e^{\bC_{\text{row}}(t-v_1)} \EE(\delta_{k_1} m_{j}^{[i,k_1]}(v_1))^2 dv_1 . 
\]
We now similarly consider 
$\delta_{k_1} m_{j}^{[i,k_1]}(v_1)$, and focus on the dependence of
this random variable on 
$\left( W_{k_1 k_2}(t) \right)_{k_2\neq i, k_1}$. Applying the same argument
as above to $\EE(\delta_{k_1} m_{j}^{[i,k_1]}(v_1))^2$, we obtain 
\[
\EE(\delta_{k_1} m_{j}^{[i,k_1]}(v_1))^2 \leq 
 6 s_{k_1j} t + \bC_{\text{row}} \max_{\substack{k_2 \neq i,j,k_1 \\ 
    \sigma_{k_1} = \pm 1}} 
 \int_0^{t} e^{\bC_{\text{row}}(t-v_2)} 
  \EE(\delta_{k_2} m_{j}^{[i,k_1,k_2]}(v_1,v_2))^2 dv_2 , 
\]
and thus, 
\begin{multline*} 
\EE (\delta_i m_j^{[i]})^2 
\leq  \frac{6 \bs C_s}{K_n} t \left( 1 + (e^{\bC_{\text{row}}t} -1) \right) \\ 
+ \bC_{\text{row}}^2 \max_{\substack{k_1 \neq i, j \\ 
    \sigma_i = \pm 1}} 
  \max_{\substack{k_2 \neq i,j,k_1 \\ 
    \sigma_{k_1} = \pm 1}} 
 \int_0^t \int_0^t e^{\bC_{\text{row}}(2t-v_1-v_2)} 
  \EE(\delta_{k_2} m_{j}^{[i,k_1,k_2]}(v_1,v_2))^2 dv_1 dv_2 . 
\end{multline*} 
Iterating, we end up with 
\begin{align*}
\EE (\delta_i m_j^{[i]})^2 &\leq  
\frac{6 \bs C_s}{K_n} t \left( 1 + (e^{\bC_{\text{row}}t} -1) + \cdots + 
  (e^{\bC_{\text{row}}t} -1)^{n-3} \right)  \\
&\phantom{=} + \bC_{\text{row}}^{n-2} \max_{\substack{k_1 \neq i, j \\ 
    \sigma_i = \pm 1}} \ldots 
  \max_{\substack{k_{n-2} \neq i, j, k_1, \ldots, k_{n-3} \\ 
  \sigma_{k_{n-2}} = \pm 1}} 
 \int_0^t \cdots \int_0^t e^{\bC_{\text{row}} \sum_{\ell=1}^{n-2} (t - v_\ell)} \times \\
 &\phantom{=} \quad\quad\quad\quad\quad 
  \EE(\delta_{k_{n-2}} m_{j}^{[i,k_1,\ldots, k_{n-2}]}(v_1,v_2,\ldots,
  v_{n-2}))^2 dv_1 \ldots dv_{n-2}, 
\end{align*} 
which leads to the result.

\section{Proof of Lemma~\ref{cA}} 
\label{prf-cA} 
Let us denote as $\cG_i$ the set of rooted, undirected, labelled and
connected graphs such that $\ell(\circ) = i$, $| V(G) | \leq r$, 
and the property 
\[
\{u,v \} \in E(G) \ \Rightarrow \ \ell(u) \in \cK_{\ell(v)} .
\]
We denote as $\cR$ the set of all the elements of $\cG_i$ but without the 
labels. Given a graph $G \in \cG_i$, let us denote as $\bar G = \bs U(G) \in 
\cR$ the unlabelled version of $G$.  With these notations, we have
\begin{equation}
\label{xcheck} 
 \left| \cA \right| = 
 \sum_{\bar G \in \cR} \ 
 \sum_{\substack{G \in \cG_i \, : \\\bs U(G) = \bar G}} \ 
 \left| \left\{ 
 (T_1, \ldots, T_b) \in \cA \, : \, \bs G(T_1, \ldots, T_b)= G \right\} 
    \right| . 
\end{equation} 
The summand in this expression is bounded by a constant independent of $G$. 
We need to show that 
\begin{equation}
\label{GKmu} 
 \left| \left\{ G \in \cG_i \, : \, \bs U(G) = \bar G \right\} \right|  
  \leq C K_n^{r-1} . 
\end{equation} 

Given $\bar G \in \cR$, denote as $\circ$ the root node of 
$\bar G$, write $M = | V(\bar G) | - 1$, and write
$V(\bar G) \setminus \{\degree\} = [M]$.  Recalling that $\bar G$ is connected,
let us consider a spanning tree of this graph rooted in $\circ$. Denote as
$\bs\pi(v)$ the parent of the node $v$ in this tree. Writing $j_\circ = i$,
we obtain that 
\[ 
 \left| \left\{ G \in \cG_i \, : \, \bs U(G) = \bar G \right\} \right|  
 \leq \left|\left\{ (j_1,\ldots, j_M) \in [n]^M \, : \, 
 \forall k \in [M], \ j_k \in \cK_{j_{\bs\pi(k)}}  \right\} \right| . 
\] 
Denoting as $L \subset [M]$ the set of the leaves of the spanning tree, we can 
write 
\begin{align} 
\left|\left\{ (j_1,\ldots, j_M) \in [n]^M \, : \, 
 \forall k \in [M], \ j_k \in \cK_{j_{\bs\pi(k)}}  \right\} \right| 
 &= \sum_{\substack{j_1,\ldots, j_M \in [n] \, :  \\
 \forall k \in [M], \ j_k \in \cK_{j_{\bs\pi(k)}} }} 1 \nonumber \\ 
&= 
\sum_{k \in [M] \setminus L} \ 
  \sum_{j_k \in \cK_{j_{\bs\pi(k)}}}  
 \Bigl( \sum_{p \in L} \sum_{j_p \in \cK_{j_{\bs\pi(p)}}}  
   1 \Bigr)  \nonumber \\ 
 &\leq C K_n^{|L|} 
\sum_{k \in [M] \setminus L} \ \sum_{j_k \in \cK_{j_{\bs\pi(k)}}}  1 , 
\nonumber 
\end{align} 
recalling that $|\cK_j| \leq CK_n$ for all $j$ 
and using the inequality $|L| K_n \leq K_n^{|L|}$ as soon as $K_n\geq 2$. 
If we prune the leaves of the original spanning tree, what remains is a tree
made of the nodes that constitute the first sum above plus the root node.  We
can apply the pruning operation to the new tree as above, and iterate until
exhausting all the set $[M] = V(\bar G) \setminus \{\degree\}$. This leads to 
\[
\left|\left\{ (j_1,\ldots, j_M) \in [n]^M \, : \, 
 \forall k \in [M], \ j_k \in \cK_{j_{\bs\pi(k)}}  \right\} \right| 
 \leq C K_n^{M} \leq C K_n^{r-1}, 
\]
hence Inequality~\eqref{GKmu}. 
 
It is furthermore easy to check that 
\[
\left| \cR \right| \leq C ,
\]
and the lemma is proven.  


\def\cdprime{$''$} \def\cprime{$'$} \def\cprime{$'$} \def\cprime{$'$}
  \def\cprime{$'$}

\end{document}